\documentclass[11pt]{article}

  \usepackage{smoothing_paper}
\usepackage[section]{placeins}
\usepackage{tabularx,ragged2e,booktabs,caption}
\usepackage{authblk} 
\newcommand{\ie}{\emph{i.e.}}

\newcommand{\expt}[1]{\mathrm{E}\left[#1\right]}

\newcommand{\rset}{\mathbb{R}}
\newcommand{\nset}{\mathbb{N}}

\newcommand{\PERIOD}{.}
\newcommand{\COMMA}{,}

\newcommand{\Ordo}[1]{{\mathcal{O}}\left(#1\right)}

\makeatletter
\def\BState{\State\hskip-\ALG@thistlm}
\makeatother

\title{Numerical Smoothing  with Hierarchical Adaptive Sparse Grids and
Quasi-Monte Carlo Methods for Efficient Option Pricing}

\author[1]{Christian Bayer}
\author[2]{Chiheb Ben Hammouda\thanks{benhammouda@uq.rwth-aachen.de}}
\author[3,4]{Ra\'ul Tempone}
\affil[1]{Weierstrass Institute for Applied Analysis and Stochastics (WIAS), Berlin, Germany.}
\affil[2]{Chair of Mathematics for Uncertainty Quantification, RWTH Aachen University, Aachen, Germany.}
\affil[3]{King Abdullah University of Science and Technology (KAUST), Computer, Electrical and Mathematical Sciences \& Engineering Division (CEMSE), Thuwal, Saudi Arabia.}
\affil[4]{Alexander von Humboldt Professor in Mathematics for Uncertainty Quantification, RWTH Aachen University, Aachen, Germany.}

\begin{document}
	\date{}
\maketitle

\begin{abstract}
When approximating the expectations of a functional of a solution to a stochastic differential equation, the numerical performance of deterministic quadrature  methods, such as sparse grid quadrature and quasi-Monte Carlo (QMC) methods,  may  critically depend on the regularity of the integrand.  To overcome this issue and improve   the  regularity structure of the problem,  we consider cases in which analytic  smoothing (bias-free mollification) cannot be performed and introduce a novel  numerical smoothing  approach by combining a  root-finding method with   a one-dimensional  numerical integration with respect to a single well-chosen variable.  We prove that,  under appropriate conditions, the resulting  function of the remaining  variables is  highly smooth, potentially affording the improved  efficiency of adaptive sparse grid quadrature (ASGQ) and QMC methods, particularly when  combined with hierarchical transformations (\ie,   the Brownian bridge and Richardson extrapolation on the weak error). This approach facilitates the effective treatment of high dimensionality.  Our study is motivated by option pricing problems, focusing on dynamics where  the discretization of the asset price is necessary.   Based on our analysis and  numerical experiments, we demonstrate the advantages of combining numerical smoothing with the ASGQ and QMC methods over these methods without smoothing  and  the Monte Carlo approach.   Finally, our approach is generic and can be applied to solve a broad class of problems,  particularly approximating distribution functions, computing financial Greeks, and estimating risk quantities.

\

\textbf{Keywords} Adaptive sparse grid quadrature, quasi-Monte Carlo, numerical smoothing, Brownian bridge, Richardson extrapolation, option pricing,  Monte Carlo, distribution functions, Greeks, risk estimation

\textbf{2010 Mathematics Subject Classification} 65C05,  65D30, 65D32, 65Y20, 91G20, 91G60.
\end{abstract}

\section{Introduction}
In several applications, such as pricing digital and barrier options, computing  financial Greeks, and estimating  risk quantities and  distribution functions, one is interested in efficiently computing the expectation of a functional $g$ of a solution $X$ to a stochastic differential equation (SDE):
\begin{small}
\begin{equation}\label{eq:QoI}
\expt{g(X)}.
\end{equation}\end{small}
Approximating  \eqref{eq:QoI} is usually challenging due to the   combination of two complications: 
\begin{enumerate}
\item An input space can have   high dimensionality for many reasons, including (i) the time discretization of an SDE that describes the dynamics or (ii) having numerous  underlying assets.
\item The  payoff function, $g$, exhibits  low regularity with respect to (w.r.t.) the input parameters.
\end{enumerate}
The first class of methods for approximating \eqref{eq:QoI} relies on Monte Carlo (MC) techniques. Although the convergence rate of  the  standard MC method is insensitive to  the input space  dimensionality and the regularity of the observable $g$, the convergence may be very slow. 
Moreover, it may not exploit the available regularity structure that could help achieve better convergence rates, except for  multilevel MC methods \cite{giles2015multilevel,bayer2020multilevel}, where Lipschitzity is necessary to obtain  optimal convergence rates. Another  class of methods  relies on deterministic quadrature  techniques (e.g., sparse grid quadrature  \cite{gerstner1998numerical,barthelmann2000high,bungartz2004sparse}, adaptive sparse grid quadrature (ASGQ) \cite{bayersmoothing,bayer2018hierarchical,ben2020hierarchical,bayer2022optimal}, and quasi-MC (QMC) \cite{niederreiter1992random, bayersmoothing,bayer2018hierarchical}).  In this work, we introduce a numerical smoothing technique to improve the performance of deterministic quadrature  approaches by improving  the  regularity structure of the problem.

The high dimensionality of the input space and   existence of discontinuities\footnote{We consider  discontinuities either in the gradients (kinks) or  in the function (jumps).} in the integrand  considerably degrade the performance of  deterministic  quadrature methods.  Some studies \cite{griebel2013smoothing, griebel2017note,griewank2017high,bayersmoothing,xiao2018conditional}  have addressed   cases involving integrands with discontinuities; however,  the emphasis was on the QMC method. In particular, \cite{griebel2013smoothing,griebel2017note,griewank2017high} focused   on the   theoretical aspects of employing the  QMC method in such a setting.    An adaptive version of the QMC method combined with geometric random splitting was employed for pricing multidimensional vanilla options for the Black-Scholes model \cite{deluigi:hal-00746872}.   Moreover, the low regularity of the integrand was addressed in previous studies by (i) performing bias-free mollification  using the conditional expectation over a subset of integration variables \cite{bayersmoothing,xiao2018conditional,bayer2018hierarchical}, or (ii)  mapping the problem to the frequency space  \cite{bayer2022optimal},  implying a better regularity structure compared to the physical space, when applicable.\footnote{The Fourier transform of the density function is available and inexpensive to compute.}

This work considers cases where  bias-free mollification cannot be performed. We introduce a novel numerical smoothing
technique based on (i) identifying discontinuity locations in a lower-dimensional space using hierarchical path generation and a  linear transformation of the coarsest factors,  (ii) solving  the  discontinuities using root-finding algorithms, (iii)  employing suitable transformations of the integration domain, and (iv) a  numerical preintegration step  w.r.t.~the dimension containing  discontinuities. We prove that, under appropriate conditions, the resulting function of the remaining variables is  highly smooth, potentially affording  improved efficiency of the ASGQ and QMC methods, particularly when combined with hierarchical transformations to treat the high dimensionality effectively \cite{bayersmoothing}.     Given that ASGQ and QMC methods benefit from anisotropy, the first technique involves employing  a hierarchical  path generation method  based on the Brownian bridge construction to reduce the effective dimension. The second technique  involves employing the Richardson extrapolation to reducw the bias (weak error),  subsequently  reducing the  number of time steps required at the coarsest level to achieve a certain error tolerance and    decreasing the total number of dimensions required for the integration problem. Our analysis and numerical experiments demonstrate the advantage of our approach,  substantially outperforming the  ASGQ and QMC methods without smoothing and the MC approach,  for high-dimensional examples and  dynamics where discretization is needed, such as the Heston model.

The outline of this study is as follows: Section \ref{sec:General setting} explains the technique of numerical smoothing, the selection of the optimal smoothing direction, and  the different building blocks that constitute our hierarchical quadrature methods.  Section \ref{sec:Analiticity Analysis} presents  the smoothness analysis of the resulting integrand after  numerical smoothing.  Next,     Section \ref{sec:Error discussion in the context of ASGQ method} discusses the  error and work     for the ASGQ  method  with numerical smoothing. Finally,   Section \ref{sec:Numerical experiments: Numerical smoothing with ASGQ} reports     the results of the  numerical experiments conducted using the ASGQ, QMC, and MC methods.  These results verify the considerable computational gains achieved using the ASGQ and QMC methods (both combined with numerical smoothing) over the MC method and the standard (without smoothing) ASGQ and QMC methods.

\section{Problem Setting and Approach  Formulation}\label{sec:General setting}
To demonstrate the application of our approach, we work mainly with two possible structures of the observable $g$:
\begin{equation}\label{eq:payoff structures}
\text{(i)}\:  g(\mathbf{x})=\max(\phi(\mathbf{x}),0);\:  \text{(ii)}\: g(\mathbf{x})=\mathbf{1}_{(\phi(\mathbf{x}) \ge 0)}, \: \mathbf{x} \in \rset^d,
\end{equation} 
where the function $\phi: \rset^d \mapsto \rset$ is assumed to be smooth.

We introduce the notation $\mathbf{x}_{-j}$ to denote    a vector with length $d-1$ representing all  variables other than $x_j$ in $\mathbf{x}$. Abusing the notation, we define  $\phi(\mathbf{x})=\phi(x_j,\mathbf{x}_{-j})$, and  for  ease of presentation, we assume that, for fixed $\mathbf{x}_{-j}$,  the function $\phi(x_j,\mathbf{x}_{-j})$ either has a simple root or is positive for all $x_j \in \rset$. This is guaranteed by the monotonicity condition \eqref{assump:Monotonicity condition} and    infinite growth condition \eqref{assump:Growth condition}, which  are assumed  for some $j \in \{1,\dots,d\}$.
\begin{align}
	\frac{\partial \phi}{\partial x_j}(\mathbf{x}) &>0,\: \forall \mathbf{x} \in \rset^d \: \: \textbf{(Monotonicity condition)\footnotemark}  \label{assump:Monotonicity condition}\\
	\underset{x_j \rightarrow +\infty}{\lim} \phi(\mathbf{x})&=\underset{x_j \rightarrow +\infty}{\lim} \phi(x_j,\mathbf{x}_{-j})=+\infty, \: \forall \mathbf{x}_{-j} \in \rset^{d-1}\: \text{or} \:\: \frac{\partial^2 \phi} {\partial x_j^2}(\mathbf{x}) \ge 0,\: \forall \mathbf{x} \in \rset^{d}  \: \: \textbf{(Growth condition)}  \label{assump:Growth condition} \PERIOD
\end{align}
\footnotetext{We present the monotonicity condition for an increasing function without  loss of generality.  However, the assumption still holds for a decreasing function, which may be the case when considering a spread option.}
Our approach can be easily extended to the case of finitely many roots without accumulation. We explain this extension in Remark \ref{rem: Generalization of numerical smoothing to the case of countable multiple roots}.
\subsection{Continuous-time formulation and optimal smoothing direction}\label{sec:Continuous time formulation and optimal smoothing direction}
In this section, we characterize the optimal  smoothing direction using the continuous-time formulation.  The purpose of this work is to approximate $\expt{g(\mathbf{X}_T)}$ at final time $T$, where $g$ is a  low-regular payoff function  and $\mathbf{X}:=(X^{(1)}, \dots, X^{(d)})$ is described using the following SDE:\footnote{We assume that  $\{W^{(j)}\}_{j=1}^d$ are uncorrelated and the correlation terms are included  in the diffusion terms $b_{ij}$.}
\begin{equation}\label{eq:SDE_interest}
dX_{t}^{(i)}=a_i(\mathbf{X}_t) dt + \sum_{j=1}^d b_{ij}(\mathbf{X}_t) dW_{t}^{(j)} \PERIOD
\end{equation}
First, we hierarchically represent  $\mathbf{W}:=(W^{(1)},\dots, W^{(d)})$ as follows:
\begin{equation}\label{eq:bridge_construction}
W^{(j)}(t)= \frac{t}{T} W^{(j)}(T)+B^{(j)}(t)= \frac{t}{\sqrt{T}} Z_j+B^{(j)}(t), \: 1 \le j \le d \COMMA
\end{equation}
where $\{Z_j\}_{j=1}^d$  are  independent and identically distributed (i.i.d.) standard Gaussian  random variables (rdvs), and  $\{B^{(j)}\}_{j=1}^d$ are independent Brownian bridges.

We can hierarchically represent $\mathbf{Z}:=(Z_1,\dots,Z_d)$  as
\begin{equation*}\label{eq:smoothing_decomposition}
\mathbf{Z}= \underset{\text{One dimensional projection}}{\underbrace{P_0 \mathbf{Z}}} +  \underset{\text{Projection on the complementary}} {\underbrace{P_{\perp} \mathbf{Z}}}\COMMA
\end{equation*} 
where $P_0 \mathbf{Z}:=(\mathbf{Z}, \mathbf{v}) \mathbf{v}$, with  $|| \mathbf{v} ||=1$, and  $Z_v:=(\mathbf{Z}, \mathbf{v})$ is a standard Gaussian rdv.\footnote{The notation   $(.,.)$  denotes the scalar product operator.} 

Furthermore,  defining $ \mathbf{w}:=\mathbf{Z}- Z_v \mathbf{v}$ yields
\begin{equation}\label{eq:smoothing_decomposition_componentwise}
Z_j=Z_v v_j+ (P_{\perp} \mathbf{Z})_j= Z_v v_j+ w_j, \: 1\le j \le d.
\end{equation}
Using  \eqref{eq:bridge_construction}  and  \eqref{eq:smoothing_decomposition_componentwise} in  \eqref{eq:SDE_interest} implies that
\begin{equation}\label{eq:SDE_decomposition_componentwise_exapanded}
dX_{t}^{(i)}=\left(a_i(\mathbf{X}_t)+\sum_{j=1}^d b_{ij}(\mathbf{X}_t)  \frac{Z_v v_j}{\sqrt{T}}\right) dt+\left(\sum_{j=1}^d b_{ij}(\mathbf{X}_t) \frac{w_j}{\sqrt{T}}\right) dt +\sum_{j=1}^d b_{ij}(\mathbf{X}_t) dB_{t}^{(j)}\PERIOD
\end{equation}
If we define $H_{\mathbf{v}}\left(Z_{v}, \mathbf{w} \right):=g\left(\mathbf{X}(T)\right)$, then   \eqref{eq:smoothing_decomposition_componentwise} and  \eqref{eq:SDE_decomposition_componentwise_exapanded} can be used to yield
\begin{align}
\expt{g\left(\mathbf{X}(T)\right)}&=\expt{\expt{H_{\mathbf{v}}\left(Z_{v}, \mathbf{w} \right)\mid \mathbf{w}}}\label{eq:expect_tower_prop},\\
\text{Var}\left[g\left(\mathbf{X}(T)\right)\right]&=\expt{\text{Var}\left[H_{\mathbf{v}}\left(Z_{v}, \mathbf{w} \right)\mid \mathbf{w} \right]}+\text{Var}\left[\expt{H_{\mathbf{v}}\left(Z_{v}, \mathbf{w}\right)\mid \mathbf{w}}\right] \label{eq:variance_tower_prop}.
\end{align}
Using  \eqref{eq:expect_tower_prop} and \eqref{eq:variance_tower_prop}, the optimal smoothing direction is    characterized as the one that maximizes the smoothing effect at $T$, that is, $\mathbf{v}$ solves the following equivalent optimization problem:
\begin{equation}\label{eq:smoothing_opt_problem}
 \underset{\underset{|| \mathbf{v} || =1}{\mathbf{v} \in \rset^d}}{\min} \: \text{Var}\left[\expt{H_{\mathbf{v}}\left(Z_{v}, \mathbf{w}\right)\mid \mathbf{w}}\right] \iff \underset{\underset{|| \mathbf{v} || =1}{\mathbf{v} \in \rset^d}}{\max} \: \expt{\text{Var}\left[H_{\mathbf{v}}\left(Z_{v}, \mathbf{w}\right)\mid \mathbf{w} \right]} .
\end{equation}
The left-hand side of \eqref{eq:smoothing_opt_problem}  corresponds to reducing the variance of the original estimator by  conditioning  w.r.t.~a specific subset of rdvs, where  the best conditioning direction (i.e., leading to the least variance) depends on the   choice of $\mathbf{v}$.   Moreover,  because $\expt{g\left(\mathbf{X}(T)\right)}$ is constant,  the right-hand side of \eqref{eq:smoothing_opt_problem}   can be motivated as follows: $H_{\mathbf{v}}\left(Z_{v}, \mathbf{w}\right)$ can be understood as convoluting  $g$ with a Gaussian density whose scale parameter depends on the choice of $\mathbf{v}$.  A larger scale parameter for the corresponding Gaussian density results in better regularity for the resulting function.

Solving  \eqref{eq:smoothing_opt_problem} is difficult, and $\mathbf{v}$ is dependent on the problem. In this work, we aim to heuristically determine $\mathbf{v}$  by considering the structure of the problem. In the following section, we provide more insight on selecting $\mathbf{v}$ and performing numerical smoothing in the time-stepping setting.

\subsection{Motivation and idea of  numerical smoothing}\label{sec:Discrete time, practical motivation}
We consider $\mathbf{X}$  the solution of the SDE \eqref{eq:SDE_interest}. To illustrate our numerical smoothing idea, we consider, for  ease of presentation,   the discretized $d$-dimensional   geometric Brownian motion (GBM) model  given by\footnote{For  ease of presentation, we set the drift term  in \eqref{eq:lognormal model} to $0$.}   
\begin{equation}\label{eq:lognormal model}
	dX^{(j)}_t=\sigma^{(j)} X^{(j)}_t dW^{(j)}_t,\quad 1 \le j \le d,
\end{equation}
where $\{W^{(1)}, \dots,W^{(d)}\}$ are correlated Brownian motions with correlations $\rho_{ij}$, and $\{\sigma^{(j)}\}_{j=1}^d$ denote  the volatilities of the different assets.
  
We denote by $(Z_1^{(j)},\dots,Z_N^{(j)})$ the $N$ standard Gaussian independent rdvs that will be used to construct the approximate path of the $j$-th asset $\bar{X}^{(j)}$, where $N$ represents the number of time steps ($\Delta t=\frac{T}{N}$).  Moreover, we denote by  $\psi^{(j)}: (Z_1^{(j)},\dots,Z_N^{(j)}) \mapsto (B_1^{(j)},\dots,B_N^{(j)})$  the mapping of the Brownian bridge construction, and  by $\Phi: \left(\Delta t, \mathbf{B}\right) \mapsto \bar{\mathbf{X}}^{\Delta t}(T)$ 
 the mapping  of the time-stepping scheme, where $\mathbf{B}:=\left(B^{(1)}_1,\dots,B^{(1)}_N,\dots, B^{(d)}_1,\dots,B^{(d)}_N\right)$ is the noncorrelated  Brownian bridge\footnote{Without loss of generality,   the   correlated Brownian bridge can be obtained  via simple matrix multiplication.} and  $\bar{\mathbf{X}}^{\Delta t}(T):= \left(\bar{X}_T^{(1)}, \dots,\bar{X}_T^{(d)} \right)$. Then, the option price can be expressed as
\begin{align}\label{eq: option price as integral_basket}
	\expt{g(\mathbf{X}(T))}&\approx	\expt{g\left(\bar{X}_T^{(1)}, \dots,\bar{X}_T^{(d)} \right)}=\expt{g(\bar{\mathbf{X}}^{\Delta t}(T))}  \nonumber\\
	&=\expt{g\circ \Phi  \left(B^{(1)}_1,\dots,B^{(1)}_N,\dots, B^{(d)}_1,\dots,B^{(d)}_N \right)} \nonumber\\
		&=\expt{g\circ\Phi  \left(\psi^{(1)}(Z_1^{(1)}, \dots, Z_N^{(1)}), \dots, \psi^{(d)}(Z_1^{(d)},\dots,Z^{(d)}_N)\right)} \nonumber\\
	&=\int_{\rset^{d \times N}} G(z_1^{(1)}, \dots, z_N^{(1)}, \dots, z_1^{(d)},\dots,z^{(d)}_N)) \rho_{d \times N}(\mathbf{z}) dz_1^{(1)} \dots dz_N^{(1)} \dots z_1^{(d)} \dots dz^{(d)}_N \COMMA
\end{align}
where\footnote{The formulation of our method is generic; for instance the mapping $\psi^{j}$  may be based on Haar basis functions as in \eqref{eq:Brownian-motion-truncated} instead of the Brownian bridges. Moreover,  a different scheme for the mapping $\Phi$ may be considered instead of the  Euler–Maruyama scheme  used in this work.}  $G:=g  \circ \Phi  \circ \left(\psi^{(1)}, \dots, \psi^{(d)}\right)$ and  $\rho_{d \times N}$ represents the $d \times N$  multivariate Gaussian density.

Moreover,   the numerical approximation of $X^{(j)}(T)$, using the Euler–Maruyama scheme, satisfies
\begin{equation}\label{eq:discrete_rep}
	\bar{X}^{(j)}(T)=X_0^{(j)} \prod_{n=0}^{N-1} \underset{:=f_n^{(j)}(Z^{(j)}_1; \mathbf{Z}^{(j)}_{-1})}{\underbrace{\left[ 1+\frac{\sigma^{(j)}}{\sqrt{T}} Z^{(j)}_1 \Delta t+ \sigma^{(j)} \Delta B^{(j)}_{n}\right]}}, \quad 1 \le j \le d \COMMA
\end{equation}
where $\Delta B^{(j)}_{n}:= B^{(j)}_{n+1}- B^{(j)}_{n}$.
\begin{remark}
Equation \eqref{eq:discrete_rep} holds even for stochastic volatility models, where $\sigma^{(j)}$  is a nonconstant and changes at each time step.
\end{remark}
\subsubsection{Step $1$ of numerical smoothing: Root finding for  the discontinuity location}\label{sec:Step $1$: Numerical smoothing}
In this step, the discontinuity location is determined by solving the corresponding root-finding problem in one dimension  after adopting suboptimal linear mapping for the coarsest factors of the Brownian increments $\mathbf{Z}_1:=(Z^{(1)}_1 , \dots, Z^{(d)}_1)$:
\begin{equation}\label{eq:linear_transformation}
\mathbf{Y}=\mathcal{A} \mathbf{Z}_1 \COMMA
\end{equation}
where $\mathcal{A} $  is a  $d \times d $ matrix representing a linear mapping.  To connect with Section \ref{sec:Continuous time formulation and optimal smoothing direction}, the smoothing direction $\mathbf{v}$ 
is expressed using the first row of  $\mathcal{A}$, which  is generally orthogonal, selected from a family of rotations. For instance, if we consider an arithmetic basket call option,   a sufficiently suitable selection of   $\mathcal{A}$ is  a rotation matrix, with the first  row (corresponding to the smoothing direction $\mathbf{v}$ introduced in Section \ref{sec:Continuous time formulation and optimal smoothing direction}) leading to $Y_1=\sum_{i=1}^d Z_1^{(i)}$ up to rescaling without any constraint for the remaining rows.  In practice, we construct $\mathcal{A}$ by fixing the first row to\footnote{Note that $\mathbf{1}_{1 \times d}$ denotes the row vector with dimension $d$,  where all its coordinates are $1$.}  $\frac{1}{\sqrt{d}} \mathbf{1}_{1 \times d}$,  and the remaining rows are obtained using the Gram-Schmidt procedure.

From \eqref{eq:discrete_rep},  using \eqref{eq:linear_transformation}, we obtain
 \begin{equation*}
 	\bar{X}^{(j)}(T)= X_0^{(j)} \prod_{n=0}^{N-1} f_n^{(j)}\left((\mathcal{A}^{-1} \mathbf{Y})_{j};  \mathbf{Z}^{(j)}_{-1} \right)=X_0^{(j)} \prod_{n=0}^{N-1} F_n^{(j)}(Y_{1};\mathbf{Y}_{-1},   \mathbf{Z}^{(j)}_{-1}), \quad 1 \le j \le d,
 \end{equation*}
where, by defining $\mathcal{A}^{\text{inv}}:= \mathcal{A}^{-1}$, we have
\begin{equation*}
F_n^{(j)}(Y_1;\mathbf{Y}_{-1}, \mathbf{Z}^{(j)}_{-1}) =  \left[ 1+\frac{\sigma^{(j)} \Delta t}{\sqrt{T}} A^{\text{inv}}_{j1} Y_1 +\frac{\sigma^{(j)}}{\sqrt{T}} \left( \sum_{i=2}^d A^{\text{inv}}_{ji} Y_i  \right) \Delta t+ \sigma^{(j)} \Delta B^{(j)}_{n}\right] \PERIOD
\end{equation*}
Considering that  the irregularity  is located at $\phi(\bar{\mathbf{X}}^{\Delta t}(T))=0$ (see \eqref{eq:payoff structures})\footnote{The locations may differ depending on the considered payoff function; for instance, many payoffs in quantitative finance have kinks  at the strike price.} then  to determine the discontinuity location  $y^{\ast}_1:=y^{\ast}_1(\mathbf{y}_{-1}, \mathbf{z}^{(1)}_{-1},\dots,\mathbf{z}^{(d)}_{-1} )$, we must find, for fixed $\mathbf{y}_{-1}$, $\mathbf{z}^{(1)}_{-1},\dots,\mathbf{z}^{(d)}_{-1} $, the roots of  $P(y^\ast_1)$:
\begin{equation}\label{eq:roots_function}
		\phi(\bar{\mathbf{X}}^{\Delta t}(T))=   \phi \left(  X_0^{(1)}  \prod_{n=0}^{N-1}  F_n^{(1)}(y^{\ast}_1;\mathbf{y}_{-1},  \mathbf{z}^{(1)}_{-1}), \dots , X_0^{(d)}  \prod_{n=0}^{N-1}  F_n^{(d)}(y^{\ast}_1;\mathbf{y}_{-1},  \mathbf{z}^{(d)}_{-1}) \right) := P(y^\ast_1)=0.
	\end{equation}
We use the Newton iteration method to determine the approximated discontinuity location,   $\bar{y}^\ast_1:=\bar{y}^\ast_1(\mathbf{y}_{-1}, \mathbf{z}^{(1)}_{-1},\dots,\mathbf{z}^{(d)}_{-1} )$.
\begin{remark}
		We chose $\mathbf{Z}_1$ for the numerical smoothing direction in  \eqref{eq:linear_transformation} for two reasons: (i) in this work, we consider European options whose payoff functionals depend only  on the assets prices at the final time $T$; and  (ii) the Brownian bridge construction creates a hierarchy of importance for the  rdvs  such that $\mathbf{Z}_1$ tends to be  the random factor most contributing to the information in $\mathbf{X}(T)$.    It may be more appropriate to consider a linear combination of $Z^{(j)}_1, \dots,Z^{(j)}_N$, $1 \le j \le d$ for the smoothing direction when considering payoff functionals that depend on the whole path of the asset price, such as  Asian options. Investigating this possibility is left for future work. Finally, the selection of $\mathcal{A}$ creates a new hierarchy of smoothness that depends more on the payoff structure.
\end{remark}
\begin{remark}
We recall that $\mathcal{A}$ is generally selected from a family of rotations, depending on the problem and payoff structure.  The investigation of optimal choices of  $\mathcal{A}$ for various settings  is  left for  future work, where we intend to perform a sensitivity analysis regarding possible choices.
	\end{remark}
\subsubsection{Step $2$ of numerical smoothing: Numerical preintegration}\label{sec:Step $2$: Integration}
In this stage, we perform the numerical preintegrating step w.r.t.~the direction considered for finding the root to determine $y^\ast_1$. Using  Fubini's theorem and    \eqref{eq: option price as integral_basket}, we obtain
\begin{align}\label{eq: pre_integration_step_wrt_y1_basket}
	\expt{g(\mathbf{X}(T))}\approx	\expt{g\left(\bar{X}_T^{(1)}, \dots,\bar{X}_T^{(d)} \right)} 
&:=\expt{I\left(\mathbf{Y}_{-1}, \mathbf{Z}^{(1)}_{-1},\dots,\mathbf{Z}^{(d)}_{-1} \right)}\\ \nonumber
	&\approx\expt{\bar{I}\left(\mathbf{Y}_{-1}, \mathbf{Z}^{(1)}_{-1},\dots,\mathbf{Z}^{(d)}_{-1} \right)} \COMMA
\end{align}
where
\begin{align}\label{eq:smooth_function_after_pre_integration}
 I\left(\mathbf{y}_{-1},\mathbf{z}^{(1)}_{-1},\dots,\mathbf{z}^{(d)}_{-1}\right)&=\int_{\rset} G\left(y_1,\mathbf{y}_{-1},\mathbf{z}^{(1)}_{-1},\dots,\mathbf{z}^{(d)}_{-1} \right) \rho_{1}(y_1) dy_1 \nonumber\\
 &= \int_{-\infty}^{y^\ast_1} G\left(y_1,\mathbf{y}_{-1},\mathbf{z}^{(1)}_{-1},\dots,\mathbf{z}^{(d)}_{-1} \right) \rho_{1}(y_1) dy_1+ \int_{y_1^\ast}^{+\infty} G\left(y_1,\mathbf{y}_{-1},\mathbf{z}^{(1)}_{-1},\dots,\mathbf{z}^{(d)}_{-1} \right) \rho_{1}(y_1) dy_1,
\end{align}
and $\bar{I}$ is the  approximation of $I$  obtained using  Newton iteration and  the two-sided Laguerre quadrature rule,  expressed as
\begin{equation}\label{eq:expression_h_bar}
	\bar{I}(\mathbf{y}_{-1},\mathbf{z}^{(1)}_{-1},\dots,\mathbf{z}^{(d)}_{-1}  ) \coloneqq \sum_{k=0}^{M_{\text{Lag}}} \eta_k \; G\left( \zeta_k\left(\bar{y}^{\ast}_1\right),
	\mathbf{y}_{-1},\mathbf{z}^{(1)}_{-1},\dots,\mathbf{z}^{(d)}_{-1}   \right),
\end{equation}
where $\bar{y}^{\ast}_1$ denotes the approximated discontinuity location and $M_{\text{Lag}}$ represents the number of Laguerre quadrature points  $\zeta_k \in \R$ with $\zeta_0 = \bar{y}^{\ast}_1$ and corresponding weights $\eta_k$.\footnote{Of course, the points $\zeta_k$ must be selected in a
	systematic manner depending on $\bar{y}^{\ast}_1$.}

The numerical smoothing treatment enables us to obtain a highly smooth integrand $\bar{I}$ (see   Section \ref{sec:Analiticity Analysis} for the smoothness analysis).
\begin{remark}[Extending the numerical smoothing idea to other payoffs and dynamics]
	Although we consider the case of the multivariate GBM model to illustrate our numerical  smoothing approach, we believe that this concept  is generic and can be extended straightforwardly to several types of payoff functions and dynamics  (see Section \ref{sec:Numerical experiments: Numerical smoothing with ASGQ} for  different tested examples).
\end{remark} 
\begin{remark}[Extending the numerical smoothing approach to the case of  multiple roots]\label{rem: Generalization of numerical smoothing to the case of countable multiple roots}
The aforementioned preintegration step can be generalized when finitely many discontinuities exist without accumulation, occurring  either because of   the payoff structure or  the use of the Richardson extrapolation. If we have $R$  multiple roots,  $\{y_i^\ast\}_{i=1}^{R}$ with the following order $y^\ast_1<y^\ast_2<\dots<y^\ast_{R}$,  the smoothed integrand in  \eqref{eq: pre_integration_step_wrt_y1_basket} is expressed as follows:
		\begin{align}\label{eq:smooth_function_after_pre_integration_multiple_root}
			I\left(\mathbf{y}_{-1},\mathbf{z}^{(1)}_{-1},\dots,\mathbf{z}^{(d)}_{-1}\right)&=\int_{\rset} G\left(y_1,\mathbf{y}_{-1},\mathbf{z}^{(1)}_{-1},\dots,\mathbf{z}^{(d)}_{-1} \right) \rho_{1}(y_1) dy_1 \nonumber\\
			=& \int_{-\infty}^{y^\ast_1} G\left(y_1,\mathbf{y}_{-1},\mathbf{z}^{(1)}_{-1},\dots,\mathbf{z}^{(d)}_{-1} \right) \rho_{1}(y_1) dy_1 \nonumber\\
			& + \sum_{i=1}^{R-1} \int_{y^\ast_{i}}^{y^\ast_{i+1}} G\left(y_1,\mathbf{y}_{-1},\mathbf{z}^{(1)}_{-1},\dots,\mathbf{z}^{(d)}_{-1} \right) \rho_{1}(y_1) dy_1\nonumber\\
			& + \int_{y_{R}^\ast}^{+\infty} G\left(y_1,\mathbf{y}_{-1},\mathbf{z}^{(1)}_{-1},\dots,\mathbf{z}^{(d)}_{-1} \right) \rho_{1}(y_1) dy_1,
		\end{align}
and its approximation $\bar{I}$   is given by
			\begin{align*}\label{eq:expression_h_bar_multiple_root}
				\bar{I}(\mathbf{y}_{-1},\mathbf{z}^{(1)}_{-1},\dots,\mathbf{z}^{(d)}_{-1}  ) \coloneqq  &\sum_{k=0}^{M_{\text{Lag},1}} \eta^{\text{Lag}}_k \; G\left( \zeta^{\text{Lag}}_{k,1}\left(\bar{y}^{\ast}_1\right),\mathbf{y}_{-1},\mathbf{z}^{(1)}_{-1},\dots,\mathbf{z}^{(d)}_{-1}   \right) \nonumber\\
				&+ 	\sum_{i=1}^{R-1} \left( \sum_{k=0}^{M_{\text{Leg},i}} \eta^{\text{Leg}}_k  \; G\left( \zeta^{\text{Leg}}_{k,i}\left(\bar{y}^{\ast}_{i}, \bar{y}^{\ast}_{i+1}\right),
				\mathbf{y}_{-1},\mathbf{z}^{(1)}_{-1},\dots,\mathbf{z}^{(d)}_{-1}   \right) \right)\nonumber\\
				&+ \sum_{k=0}^{M_{\text{Lag},R}} \eta^{\text{Lag}}_k \; G\left( \zeta^{\text{Lag}}_{k,R}\left(\bar{y}^{\ast}_{R}\right),
				\mathbf{y}_{-1},\mathbf{z}^{(1)}_{-1},\dots,\mathbf{z}^{(d)}_{-1}   \right), \nonumber
			\end{align*}
		where $\{\bar{y}_i^\ast\}_{i=1}^{R}$ are the approximated discontinuity locations,  $M_{\text{Lag},1}$ and $M_{\text{Lag},R}$ are the number of Laguerre quadrature points  $ \zeta^{\text{Lag}}_{.,.} \in \R$  with corresponding weights $\eta^{\text{Lag}}_{.} $, and  $\{M_{\text{Leg},i}\}_{i=1}^{R-1}$  are the numbers of Legendre quadrature points  $ \zeta^{\text{Leg}}_{.,.}$ with corresponding weights $\eta^{\text{Leg}}_{.}$\footnote{The points $\zeta^{\text{Lag}}_{.,.}$ and $\zeta^{\text{Leg}}_{.,.}$  must be selected systematically depending on $\{\bar{y}^{\ast}_i\}_{i=1}^{R}$.}. Moreover, $\bar{I} $ can be  approximated further depending on  the decay of  $ G\times \rho_{1} $ in the semi-infinite domains in \eqref{eq:smooth_function_after_pre_integration_multiple_root}  and   how close the roots are.  This approximation enables  dealing  with a countable number of discontinuities by keeping them toward infinity and then truncating the domain.
\end{remark}
\begin{remark}
		Our approach can be extended to solve a broad class of problems, particularly for  estimating risk quantities and computing Greeks (sensitivities) for discontinuous financial  payoffs (e.g.,  for the low-factor London interbank offer rate (LIBOR) market \cite{fries2008conditional, chan2013fast}).  The numerical smoothing idea  can be used to design novel efficient methods based on the pathwise approach   relying on the differentiability of the payoff.   We intend to explore these directions in the future, where we should also address  additional  challenges related to pathwise simulation.
\end{remark}
\begin{example}[Simple illustration: Single digital call under  the GBM model]\label{rem:Single digital call under GBM model}
	We let $dX = \sigma XdW$, with $\sigma>0$, and $W$ is a Brownian motion.  We  consider   $g$  given by (ii) in \eqref{eq:payoff structures},  where  $\phi(x) = x-K$, and $K$ is the strike price. From previous section,  the discontinuity is located at $y_1^{\ast}$,  which  is   an invertible function satisfying $X(T; y_1^\ast(x), \mathbf{z}_{-1})=x$. In this particular case,    $y_1^{\ast}$ is deterministic (it does not depend on the Brownian bridge increments) and is given by
	$$ y_1^{\ast} = (\log(K/X_0) + T \sigma^2/2) \frac{1}{\sqrt{T} \sigma}.$$
	Then,
	\begin{equation}
	I(\mathbf{z}_{-1})=\int_{\rset}  \mathbf{1}_{\{X(T; y,\mathbf{z}_{-1})>K\}}  \frac{1}{\sqrt{2 \pi}} \exp(-y^2/2) dy= \mathbb{P}\left( Y> y^\ast(K)\right),
	\end{equation}
where $Y \sim \mathcal{N} (0,1).$ Generally, $y_1^{\ast}$ is not deterministic; for instance when considering the Heston model (see \eqref{eq:dynamics Heston})  because the volatility is stochastic.
	\end{example}

\subsection{Hierarchical quadrature methods combined with numerical smoothing}\label{sec:Details of our approach}
After performing the numerical smoothing step, we end up with an  integration problem \eqref{eq: pre_integration_step_wrt_y1_basket} of a highly regular integrand $\bar{I}$ in a $(dN-1)$-dimensional space (see Section \ref{sec:Analiticity Analysis} for the regularity analysis).    The second  stage of our approach involves  approximating    \eqref{eq: pre_integration_step_wrt_y1_basket} efficiently. Thus, we employ   the ASGQ  and QMC methods (see Section \ref{sec:Brief description of ASGQ and QMC methods} for a brief description of the methods in our context; refer to \cite{bayer2018hierarchical,bayer2022optimal} for more details).  In general, given \eqref{eq:roots_function},   we must to compute this root for any quadrature or  QMC point.   However,  there are cases in which the root is deterministic (see Remark \ref{rem:Single digital call under GBM model}), where it is only computed once for all quadrature or QMC points.

The dimension of the integration problem \eqref{eq: pre_integration_step_wrt_y1_basket}  may   become very large because of  (i)  numerous time steps  $N$  in the discretization scheme or (ii) a large  number of  assets, $d$. To overcome the high-dimensionality issue, we apply an idea similar to that introduced in \cite{bayer2018hierarchical} and  combine   the ASGQ and QMC methods with  two hierarchical transformations. We first employ a hierarchical  path generation  based on the  Brownian bridge construction to reduce the effective dimension and then use the Richardson extrapolation to reduce the bias and, consequently,  the dimension of the integration problem. More details on these two hierarchical representations are available in \cite{bayer2018hierarchical}.

\subsubsection{Brief description of ASGQ and QMC methods}\label{sec:Brief description of ASGQ and QMC methods}
 We  denote  by  $\boldsymbol{\beta}:= (\beta_n)_{n=1}^{dN-1}  \in \nset^{dN-1}$  a  multi-index and   by  $\bar{I}_{\boldsymbol{\beta}}:= \sum_{j=1}^{\#\mathcal{T}^{m(\boldsymbol{\beta})}} \bar{\omega}_j \bar{I}(\hat{\mathbf{y}}_j) $   the Cartesian quadrature estimator\footnote{The 	cardinality of $\mathcal{T}^{m(\boldsymbol{\beta})}$ is  $\#\mathcal{T}^{m(\boldsymbol{\beta})}=\prod_{n=1}^{dN-1} m (\beta_n)$ with $m(\beta_n)$ quadrature points along the $n$th dimension, and $m: \nset \rightarrow \nset$ is a strictly increasing function with $m(0)=0$ and $m(1)=1$.} of $E[\bar{I}]$ in the  tensor  grid $\mathcal{T}^{m(\boldsymbol{\beta})}= \prod_{n = 1}^{dN-1}  \mathcal{H}^{m(\beta_n)}$, with  $\mathcal{H}^{m(\beta_n)}:=\{y^1_{\beta_n},\dots,y_{\beta_n}^{m(\beta_n)}\} \subset \rset$,   $\hat{\mathbf{y}}_j \in \mathcal{T}^{m(\boldsymbol{\beta})}$ are   the quadrature points and $\bar{\omega}_j$ denotes the  products of the weights of the univariate  rules.  Then, using a construction similar to that  described in    \cite{bayer2018hierarchical,bayer2022optimal},  the ASGQ estimator  for approximating \eqref{eq: pre_integration_step_wrt_y1_basket} using a set of multi-indices  $ \mathcal{I}_{\text{ASGQ}}\subset \nset^{dN-1}$ is  
		\begin{equation}\label{eq:MISC_quad_estimator}
			Q^{\text{ASGQ}}:=	Q^{\mathcal{I}_{\text{ASGQ}}} [\bar{I}]= \sum_{\boldsymbol{\beta} \in \mathcal{I}_{\text{ASGQ}}} \Delta [\bar{I}_{\boldsymbol{\beta}}],
	\end{equation}
	with 
		\begin{equation*}
			 \Delta [\bar{I}_{\boldsymbol{\beta}}] = \left( \prod_{i=1}^{dN-1} \Delta_i \right) \bar{I}_{\boldsymbol{\beta}}, \quad \text{and} \quad 	\Delta_i \bar{I}_{\boldsymbol{\beta}} := \left\{ 
			\aligned 
			\bar{I}_{\boldsymbol{\beta}} &- \bar{I}_{\boldsymbol{\beta}'}  \text{, with } \boldsymbol{\beta}' =\boldsymbol{\beta} - \mathbf{e}_i, \text{ if } \boldsymbol{\beta}_i>0 \COMMA \\
			\bar{I}_{\boldsymbol{\beta}} &, \quad  \text{ otherwise,}
			\endaligned
			\right.
	\end{equation*}
where $\mathbf{e}_i$ denotes the $i$th $(dN-1)$-dimensional unit vector.  

The construction of $\mathcal{I}_{\text{ASGQ}}$ is done a posteriori  and adaptively by profit thresholding, such that $ \mathcal{I}_{\text{ASGQ}}=\{\boldsymbol{\beta} \in \mathbb{N}^{dN-1}_{+}: P_{\boldsymbol{\beta}}	 \ge \overline{T}\}$, where  $P_{\boldsymbol{\beta}}= \frac{\abs{\Delta E_{\boldsymbol{\beta}}}}{\Delta\mathcal{W}_{\boldsymbol{\beta}}}$ is  the profit of a hierarchical surplus,  and 
\begin{align}	\label{error_contr}
	\Delta E_{\boldsymbol{\beta}} &=\left|Q^{\mathcal{I_{\text{ASGQ}}} \cup\{\boldsymbol{\beta}\}}-Q^{\mathcal{I_{\text{ASGQ}}}}\right|  \quad (\text{the error contribution})\\
	\Delta \mathcal{W}_{\boldsymbol{\beta}} &=\operatorname{Work}\left[Q^{\mathcal{I_{\text{ASGQ}}}\cup\{\boldsymbol{\beta}\}}\right]- \operatorname{Work}\left[Q^{\mathcal{I_{\text{ASGQ}}}}\right] \quad   (\text{the work contribution}).\nonumber
\end{align}	
We also use the randomized QMC (rQMC) method based on lattice rules \cite{SLOAN1985131,nuyens2014construction}, as described in  Section 4.2 in \cite{bayer2018hierarchical}. The rQMC estimator is defined as follows:
\begin{align}
	Q^{\text{rQMC}}=\frac{1}{q}\sum_{i=0}^{q-1}\left(\frac{1}{n}\sum_{k=0}^{n-1} (\bar{I} \circ F^{-1})  \left( \frac{k\mathbf{u}+\Delta^{(i)}  \: \text{mod}\: n}{n}\right)  \right),
\end{align}
where $\{\Delta^{(i)}\}_{i=0}^{q-1}$  are $q$   independent random shifts from the uniform distribution of $[0,1]^{dN-1}$, $\mathbf{u}=(u_1,\dots, u_{dN-1})$   is the  fixed lattice generating vector, and $ F^{-1}(\cdot)$ is the  inverse of the standard normal cumulative distribution function.  The total number of rQMC samples  is $M^{\text{rQMC}}=q \times n$.

\section{Smoothness Analysis and Error Discussion}\label{sec:Error discussion and smoothness analysis}

\subsection{Smoothness analysis}\label{sec:Analiticity Analysis}
To achieve the optimal performance of the ASGQ and QMC methods, the integrand should be highly smooth. In this section, we analyze the smoothness  of the integrand of interest after employing our numerical smoothing approach.   First, we introduce the notation and then state the smoothness theorem, Theorem \ref{thr:smoothness}.

For simplicity, we assume that we work on a fixed time interval $[0,T]$, with $T = 1$. Using the Haar mother wavelet
\begin{small}
	\begin{equation*}
		\psi(t) \coloneqq
		\begin{cases}
			1, & 0 \le t < \half, \\
			-1, & \half \le t < 1, \\
			0, & \text{else},
		\end{cases}
	\end{equation*}
\end{small}
we construct the Haar basis functions of $L^2\left([0,1] \right)$ by setting
	\begin{equation*}
		\psi_{-1}(t) \coloneqq \mathbf{1}_{[0,1]}(t);  \quad \psi_{n,k}(t) \coloneqq 2^{n/2} \psi\left( 2^n t - k \right), \quad n \in
		\N_0, \ k = 0, \ldots, 2^n-1.
	\end{equation*}
The support of $\psi_{n,k}$ is  $[2^{-n}k, 2^{-n}(k+1)]$. Moreover, we define
a grid $\mathcal{D}^n \coloneqq \Set{t^n_\ell \mid \ell = 0, \ldots, 2^{n+1}}$ by
$t^n_\ell \coloneqq \f{\ell}{2^{n+1}} T$. The Haar basis functions up to level
$n$ are piecewise constants with points of discontinuity given by
$\mathcal{D}^n$. Next, we define the antiderivatives of the Haar basis functions:
	\begin{equation*}
		\Psi_{-1}(t) \coloneqq \int_0^t \psi_{-1}(s) ds;\quad     \Psi_{n,k}(t) \coloneqq \int_0^t \psi_{n,k}(s) ds.
	\end{equation*}
For an i.i.d.~set of standard normal rdvs (\emph{coefficients})
$Z_{-1}$, $Z_{n,k}$, $n \in \N_0$, $k = 0, \ldots, 2^n-1$, we   define
the standard Brownian motion
	\begin{equation*}
		W_t \coloneqq Z_{-1} \Psi_{-1}(t) + \sum_{n=0}^\infty \sum_{k=0}^{2^n-1}
		Z_{n,k} \Psi_{n,k}(t),
	\end{equation*}
and the truncated version
	\begin{equation}  \label{eq:Brownian-motion-truncated}
		W_t^N \coloneqq Z_{-1} \Psi_{-1}(t) + \sum_{n=0}^N \sum_{k=0}^{2^n-1},
		Z_{n,k} \Psi_{n,k}(t).
\end{equation}
where   $W^N$ already coincides with $W$ along the grid $\mathcal{D}^N$. We
define the corresponding increments for any function or process $F$ as follows:
	\begin{equation*}
		\Delta^N_\ell F \coloneqq F(t^N_{\ell+1}) - F(t^N_\ell).
	\end{equation*}
For simplicity, we consider a one-dimensional SDE for the process $X$ as follows:
\begin{equation}
  \label{eq:SDE}
  dX_t = b(X_t) dW_t, \quad X_0 = x \in \R.
\end{equation}
We assume that $b$  and its derivatives for all
orders are bounded. Recall that we want to compute, for  $g: \R \to \R$, which is not necessarily smooth,  $\expt{g\left( X_T \right)}$. Furthermore, we  define the solution of the Euler–Maruyama  scheme along  $\mathcal{D}^N$
by $X^N_0 \coloneqq X_0 = x$; for convenience, we also define $X^N_T \coloneqq X^N_{2^N}$.
\begin{equation}\label{eq:SDE_descretized}
  X^N_{\ell+1} \coloneqq X^N_\ell + b\left( X^N_{\ell} \right) \Delta^N_\ell W, \quad \ell = 0, \ldots, 2^{N}-1.
\end{equation}
The rdv $X^N_\ell$ is a deterministic function of the rdvs $Z_{-1}$ and $\textbf{Z}^N \coloneqq \left(Z_{n,k} \right)_{n=0,
  \ldots, N, \ k=0, \ldots 2^n-1}$. Using this notation, we write
\begin{equation}\label{eq:X_l_definition}
  X^N_\ell = X^N_\ell\left(Z_{-1}, \textbf{Z}^N \right),
\end{equation}
for the appropriate (now deterministic) map $X^N_\ell: \R \times
\R^{2^{N+1}-1} \to \R$. We  write $y \coloneqq z_{-1}$ and $\textbf{z}^N$ for the
(deterministic) arguments of the function $X^N_\ell$.\footnote{We offer a note of caution  regarding the convergence as $N \to \infty$. Although
			the sequence of random processes $X^N_\cdot$ converges to the solution
			of~\eqref{eq:SDE} (under the usual assumptions on $b$), this is not true in
			any sense for  deterministic functions.}
		
We define  the deterministic function $ H^N: \R^{2^{N+1}-1} \to \R$, expressed as follows:
\begin{equation}
  \label{eq:H-function}
  H^N(\textbf{z}^N) \coloneqq E\left[ g\left( X^N_T\left( Z_{-1}, \textbf{z}^N \right) \right) \right].
\end{equation}
Before stating the main theorem that  $H^N$  satisfies, we motivate its proof and and  underlying assumptions. We consider a mollified version $g_\delta$ of $g$ and the corresponding
	function $H^N_\delta$ (defined by replacing $g$ with $g_\delta$
	in~\eqref{eq:H-function}). Tacitly, assuming that we can interchange the
	integration and differentiation (refer to  Lemma \ref{lem:dXdZ} for justification), we achieve
	\begin{equation*}
		\f{\pa H^N_\delta(\textbf{z}^N)}{\pa z_{n,k}} = E\left[ g_\delta^\prime\left( X^N_T\left( Z_{-1},
		\textbf{z}^N \right) \right) \f{\pa X^N_T(Z_{-1}, \textbf{z}^N)}{\pa z_{n,k}}\right].
	\end{equation*}
	Multiplying and dividing by $\f{\pa X^N_T(Z_{-1}, \textbf{z}^N)}{\pa y}$ and replacing
	the expectation by an integral w.r.t.~the standard normal density, we obtain
	\begin{equation}
		\label{eq:1}
		\f{\pa H^N_\delta(\textbf{z}^N)}{\pa z_{n,k}} = \int_\R \f{\pa g_\delta\left( X^N_T
			(y, \textbf{z}^N) \right)}{\pa y} \left( \f{\pa X^N_T}{\pa y}(y, \textbf{z}^N)
		\right)^{-1} \f{\pa X^N_T}{\pa z_{n,k}}(y, \textbf{z}^N) \f{1}{\sqrt{2\pi}}
		e^{-\f{y^2}{2}} dy.
	\end{equation}
	If integration by parts is possible, we can discard the mollified version and obtain the smoothness of $H^N$ because
	\begin{equation*}
		\f{\pa H^N(\textbf{z}^N)}{\pa z_{n,k}} = -\int_\R g\left( X^N_T
		(y, \textbf{z}^N) \right) \f{\pa}{\pa y} \left[ \left( \f{\pa X^N_T}{\pa y}(y, \textbf{z}^N)
		\right)^{-1} \f{\pa X^N_T}{\pa z_{n,k}}(y, \textbf{z}^N) \f{1}{\sqrt{2\pi}}
		e^{-\f{y^2}{2}}\right] dy.
	\end{equation*}
	However, there are situations in which there may be a potential problem looming in the inverse of the
	derivative w.r.t.~$y$\footnote{As an example, let us assume that
		$X^N_T(y,\textbf{z}^N) = \cos(y) + z_{n,k}$. Then, \eqref{eq:1} is generally not
		integrable.}. This observation motivates the introduction of Assumptions \ref{ass:boundedness-derivative} and \ref{ass:boundedness-inverse}.
	\begin{notation}\label{notation}
			For sequences of rdvs $F_N$, we write  that
			$F_N = \mathcal{O}(1)$ if there exists a rdv $C$ with finite moments
			of all orders,  such that for all $N$, we have $\abs{F_N} \le C $ a.s.
	\end{notation}
	\begin{assumption}
		\label{ass:boundedness-derivative}
		There are positive rdvs $C_p$ with finite moments of all orders\footnote{It is probably difficult to argue that a deterministic constant $C$  may exist in Assumption \ref{ass:boundedness-derivative} .} such that
		\begin{equation*}
			\forall N \in \N,\ \forall \ell_1, \ldots, \ell_p \in \{0, \ldots, 2^N-1\}:\ \abs{\f{\pa^p
					X^N_T}{\pa X^N_{\ell_1} \cdots \pa X^N_{\ell_p}}} \le C_p \text{ a.s.}
		\end{equation*}
		In terms of  notation~\ref{notation}, this means that $\f{\pa^p X^N_T}{\pa
			X^N_{\ell_1} \cdots \pa X^N_{\ell_p}} = \mathcal{O}(1)$.
	\end{assumption}
	Assumption~\ref{ass:boundedness-derivative} is natural   because  it is fulfilled if  the diffusion coefficient $b(\cdot)$ is smooth. This situation is valid for many option pricing models.   Besides Assumption \ref{ass:boundedness-derivative},  we  
	make  another assumption, Assumption \ref{ass:boundedness-inverse},  which might be challenging to verify in practice for some models. In Appendix \ref{appendix:Discussion of Assumption}, we explain  cases  with sufficient conditions where this assumption is valid.
	\begin{assumption}
		\label{ass:boundedness-inverse}
		For any $p \in \N$ we obtain
		\begin{equation*}
			\left( \f{\pa X^N_T}{\pa y}\left( Z_{-1}, \textbf{Z}^N \right) \right)^{-p} = \mathcal{O}(1).
		\end{equation*}
\end{assumption}
We are now in a position to state Theorem \ref{thr:smoothness} for  $H^N$. We refer to Appendix \ref{appendix:Details for the proof of Theorem} for  its proof.
\begin{theorem}
  \label{thr:smoothness}
  Assume that  $X^N_T$, defined by  \eqref{eq:SDE_descretized} and \eqref{eq:X_l_definition}, satisfies  Assumptions~\ref{ass:boundedness-derivative}  and~\ref{ass:boundedness-inverse}. Then,  for any $p \in \N$ and indices $n_1, \ldots, n_p$ and $k_1, \ldots, k_p$
  (satisfying $0 \le k_j < 2^{n_j}$),  the function  $H^N$ defined in \eqref{eq:H-function}  satisfies the following  (with constants independent of  $n_j, k_j$)
  \begin{equation*}
    \f{\pa^p H^N}{\pa z_{n_1,k_1} \cdots \pa z_{n_p,k_p}}(\textbf{z}^N) = 
    \mathcal{O}\left( 2^{-\sum_{j=1}^p n_j/2} \right). 
  \end{equation*}
In particular, $H^N$  is of  class $C^\infty$.
\end{theorem}
\begin{remark}[Regarding the analyticity of $H^N$]
	\label{rem:analyticity}
	We  expect that $H^N$ is analytic; however, the formal proof is subtle.  In particular, our proof in Appendix \ref{appendix:Details for the proof of Theorem} relies on successively
	applying the technique of dividing by
	$\f{\pa X^N_T}{\pa y}$ and then integrating by parts. Thus, the constant in $\mathcal{O}\left( 2^{-\sum_{j=1}^p n_j/2} \right)$   
	depends on $p$ and increases in $p$. In other words,
	Theorem~\ref{thr:smoothness} should be interpreted as an assertion of the anisotropy in the variables  $z_{n,k}$ rather than a statement on the
	behavior of higher  derivatives of $H^N$. Our proof reveals that the number of summands increases as $p!$. Therefore, the statement of the theorem does not already imply
	analyticity. This problem is an artifact of our construction, and
	there is no reason to assume such  behavior in general. Finally, we expect the analyticity region to shrink as $N \rightarrow \infty$, which motivates the use of the Richardson extrapolation to keep $N$  as small as possible while achieving the desired accuracy.
\end{remark}
\begin{remark}
The analysis of the smoothness direction and sufficient conditions for  Theorem \ref{thr:smoothness} to be valid at  high dimensions is an open problem and is beyond the  study scope.
\end{remark}

\subsection{Error  and work discussion  for ASGQ  combined with numerical smoothing}\label{sec:Error discussion in the context of ASGQ method}
In this section, we   analyze the  errors  in the proposed  approach when using the ASGQ   method combined with  numerical smoothing. The error analysis of the QMC method combined with numerical smoothing is almost similar, as explained in Remark \ref{rem: QMC_smoothing_error}. Following the notation  in Sections \ref{sec:Discrete time, practical motivation} and \ref{sec:Brief description of ASGQ and QMC methods},  we obtain the following error decomposition for  the ASGQ estimator, $Q^{\text{ASGQ}}$ :
\begin{align}\label{eq: error decomposition}
\expt{g(X(T)}-Q^{\text{ASGQ}}&= \underset{\text{Error I: bias or weak error}}{\underbrace{\expt{g(X(T))}- \expt{g(\bar{\mathbf{X}}^{\Delta t}(T))}}}\nonumber\\
&+\underset{\text{Error II: numerical smoothing error}} {\underbrace{\expt{I\left(\mathbf{Y}_{-1}, \mathbf{Z}^{(1)}_{-1},\dots,\mathbf{Z}^{(d)}_{-1}\right)}- \expt{\bar{I}\left(\mathbf{Y}_{-1}, \mathbf{Z}^{(1)}_{-1},\dots,\mathbf{Z}^{(d)}_{-1} \right)}}}\nonumber\\
&+ \underset{\text{Error III: ASGQ  error}}{\underbrace{\expt{\bar{I}\left(\mathbf{Y}_{-1}, \mathbf{Z}^{(1)}_{-1},\dots,\mathbf{Z}^{(d)}_{-1}\right)}- Q^{\text{ASGQ}}}}.
\end{align} 
Because we use the Euler–Maruyama scheme to simulate asset dynamics, we achieve
\begin{equation}\label{eq:Error I order}
\text{Error I}=\Ordo{\Delta t}.
\end{equation}
We  denote  by $\text{TOL}_{\text{Newton}} $ the tolerance of the  Newton method  to approximate the discontinuity location by  finding the roots of $P(y^\ast_1)$ defined in \eqref{eq:roots_function}. Thus,  $\abs{P(\bar{y}^\ast_1)} \le \text{TOL}_{\text{Newton}} $, and using the Taylor expansion, $\left(y^\ast_1-\bar{y}^\ast_1\right)=\Ordo{\text{TOL}_{\text{Newton}}}$. Consequently,  \text{Error II}  in  \eqref{eq: error decomposition} is expressed as
%
\begin{align}\label{eq:Error II order}
\text{Error II}&:=\expt{I\left(\mathbf{Y}_{-1}, \mathbf{Z}^{(1)}_{-1},\dots,\mathbf{Z}^{(d)}_{-1}\right)}- \expt{\bar{I}\left(\mathbf{Y}_{-1}, \mathbf{Z}^{(1)}_{-1},\dots,\mathbf{Z}^{(d)}_{-1} \right)}\nonumber\\
&=   E \left[\int_{-\infty}^{y^\ast_1} G\left(y_1,\mathbf{y}_{-1},\mathbf{z}^{(1)}_{-1},\dots,\mathbf{z}^{(d)}_{-1} \right) \rho_{1}(y_1) dy_1+ \int_{y_1^\ast}^{+\infty} G\left(y_1,\mathbf{y}_{-1},\mathbf{z}^{(1)}_{-1},\dots,\mathbf{z}^{(d)}_{-1} \right) \rho_{1}(y_1) dy_1\right.\nonumber\\
&\quad \quad  \quad \quad \quad \quad  \left. - \left(\int_{-\infty}^{\bar{y}^\ast_1} G\left(y_1,\mathbf{y}_{-1},\mathbf{z}^{(1)}_{-1},\dots,\mathbf{z}^{(d)}_{-1} \right) \rho_{1}(y_1) dy_1+ \int_{\bar{y}_1^\ast}^{+\infty} G\left(y_1,\mathbf{y}_{-1},\mathbf{z}^{(1)}_{-1},\dots,\mathbf{z}^{(d)}_{-1} \right) \rho_{1}(y_1) dy_1\right)\right. \nonumber\\
&\quad \quad  \quad \quad \quad \quad  \left. + \left(\int_{-\infty}^{\bar{y}^\ast_1} G\left(y_1,\mathbf{y}_{-1},\mathbf{z}^{(1)}_{-1},\dots,\mathbf{z}^{(d)}_{-1} \right) \rho_{1}(y_1) dy_1+ \int_{\bar{y}_1^\ast}^{+\infty} G\left(y_1,\mathbf{y}_{-1},\mathbf{z}^{(1)}_{-1},\dots,\mathbf{z}^{(d)}_{-1} \right) \rho_{1}(y_1) dy_1\right)\right. \nonumber\\
&\quad \quad  \quad \quad \quad \quad  \left. - \sum_{k=0}^{M_{\text{Lag}}} \eta_k \; G\left( \zeta_k\left(\bar{y}^{\ast}_1\right),
\mathbf{y}_{-1},\mathbf{z}^{(1)}_{-1},\dots,\mathbf{z}^{(d)}_{-1}   \right)\right]\nonumber\\
& =  \Ordo{\text{TOL}_{\text{Newton}}}+\Ordo{M_{\text{Lag}}^{-s/2}},
\end{align}
where   $s>0$ is related to the degree of regularity of the integrand, $G$, w.r.t.~$y_1$.\footnote{For the  parts of the  domain separated by the discontinuity  location, the derivatives of $G$ w.r.t.~$y_1$ are bounded up to order $s$.}
  
 The first  error contribution in \eqref{eq:Error II order} originates from  the gap created by integrating  $G$ over domains separated by  the approximated discontinuity location $\bar{y}^\ast_1$ instead of   $y^\ast_1$, which is the exact location.  We consider without loss of generality that $\bar{y}^{\ast}_1 <y_1^{\ast}$, then we obtain\footnote{A similar argument holds for the part of $G$ located on the right of the discontinuity.}
 \begin{align*}
 	&E \left[\int_{-\infty}^{y^\ast_1} G\left(y_1,\mathbf{y}_{-1},\mathbf{z}^{(1)}_{-1},\dots,\mathbf{z}^{(d)}_{-1} \right) \rho_{1}(y_1) dy_1 - \int_{-\infty}^{\bar{y}^\ast_1} G\left(y_1,\mathbf{y}_{-1},\mathbf{z}^{(1)}_{-1},\dots,\mathbf{z}^{(d)}_{-1} \right) \rho_{1}(y_1) dy_1\right]\\
 	&=E \left[\int_{\bar{y}^{\ast}_1}^{y^\ast_1} G\left(y_1,\mathbf{y}_{-1},\mathbf{z}^{(1)}_{-1},\dots,\mathbf{z}^{(d)}_{-1} \right) \rho_{1}(y_1) dy_1 \right]\\
 	&=E \left[\int_{\bar{y}^{\ast}_1}^{y^\ast_1} \left(G\left(\bar{y}^{\ast}_1;. \right)+G'\left(\bar{y}^{\ast}_1;. \right) (y_1-\bar{y}^{\ast}_1)+ \Ordo{(y_1-\bar{y}^{\ast}_1)^2}\right) \rho_{1}(y_1) dy_1 \right]\\
 	&=E \left[G\left(\bar{y}^{\ast}_1;. \right) \left(F(y^{\ast}_1) -F(\bar{y}^{\ast}_1) \right) + \Ordo{\text{TOL}_{\text{Newton}}^2}\right]=\Ordo{\text{TOL}_{\text{Newton}}},
 \end{align*}
where $F(\cdot)$ is the standard normal cumulative distribution function.  For the error bound above  to hold, we assume that $G$ and all its derivatives at $\bar{y}^{\ast}_1$,  apprearing in the constant in  $\Ordo{\text{TOL}_{\text{Newton}}}$,  are integrable.
 
 The second error  contribution  in  \eqref{eq:Error II order} originates from the one-dimensional preintegration step using the Laguerre quadrature, as explained in Section \ref{sec:Step $2$: Integration}. Considering that $G$
 is a  smooth function in  parts of the integration domain separated by the discontinuity,  we achieve a spectral convergence of the quadrature \cite{mastroianni1994error}, justifying the term   $M_{\text{Lag}}^{-s/2}$. For the error bound in \eqref{eq:Error II order} to hold, we assume that the constant in $\Ordo{M_{\text{Lag}}^{-s/2}}$, which depends on $\mathbf{y}_{-1}, \mathbf{z}_{-1}^{(1)},\dots, \mathbf{z}_{-1}^{(d)}$,  is integrable.

Finally, 	considering  $M_{\text{ASGQ}}$  quadrature points used in the ASGQ method, we achieve 
\begin{equation}\label{eq:Error III, ASGQ, order}
	\text{Error III}=\Ordo{M_{\text{ASGQ}}^{-p/2}},
\end{equation}
where the bound in \eqref{eq:Error III, ASGQ, order} is justified by the analysis in  \cite{chen2018sparse,ernst2018convergence} and $p:=p\left( N,d \right)>0$ is related to the degree of regularity of  $\bar{I}$, as defined in \eqref{eq: pre_integration_step_wrt_y1_basket} and \eqref{eq:smooth_function_after_pre_integration}, in the $(dN-1)$-dimensional space.\footnote{We refer to \cite{chen2018sparse,ernst2018convergence} for a clear characterization of $p$.} In this case,   our smoothness analysis (Section \ref{sec:Analiticity Analysis}) implies that $p \gg 1$, under the assumption that $\bar{I}$ converges to $I$ for  large values of  $M_{\text{Lag}}$  and small $\text{TOL}_{\text{Newton}}$.     Nevertheless,  the optimal performance for ASGQ can deteriorate (i) if $p$ and $s$ are not sufficiently large, or (ii) due to  the adverse effect of the high dimension that may severely affect  the rates.  Finally, although we work in the preasymptotic regime (small number of time steps, $N$),  the regularity parameter $p$ may deteriorate when increasing the dimension of the integration problem by increasing $N$, justifying the use of Richardson extrapolation.

Considering  \eqref{eq: error decomposition}, \eqref{eq:Error I order}, \eqref{eq:Error II order} and \eqref{eq:Error III, ASGQ, order}, the total error estimate of our approach is
\begin{equation}\label{eq:total_error_estimate}
\mathcal{E}_{\text{total, ASGQ}}:=\expt{g(X(T)}-Q^{\text{ASGQ}}=\Ordo{\Delta t}+\Ordo{M_{\text{ASGQ}}^{-p/2}}+\Ordo{M_{\text{Lag}}^{-s/2}}+ \Ordo{\text{TOL}_{\text{Newton}}}.
\end{equation}
To achieve  optimal performance, we need to optimize  the parameters  in \eqref{eq:total_error_estimate} to satisfy a certain error tolerance, $\text{TOL}$, with the least amount of work, which  can be achieved by  solving \eqref{eq:opt_ASGQ_work}:
\begin{align}\label{eq:opt_ASGQ_work}
\begin{cases} 
\underset{\left(M_{\text{ASGQ}},M_{\text{Lag}}, \text{TOL}_{\text{Newton}}\right)}{\operatorname{min}} \: \text{Work}_{\text{ASGQ}} \propto \Delta t^{-1}  \times M_{\text{ASGQ}} \times M_{\text{Lag}}  \\ 
s.t. \:   \mathcal{E}_{\text{total,ASGQ}}=\text{TOL}.
\end{cases}
\end{align}
We do not solve \eqref{eq:opt_ASGQ_work} in our experiments in Section \ref{sec:Numerical experiments: Numerical smoothing with ASGQ} (we select the parameters heuristically to achieve a suboptimal performance). However, in Appendix \ref{appendix:More details on the error and work discussion of ASGQ method combined with numerical smoothing}  we reveal  that, for a given error tolerance $\text{TOL}$,  under certain conditions of the regularity parameters $s$ and $p$ ($p,s \gg 1$), a lower bound on the computational work of the ASGQ method is of order $\text{Work}_{\text{ASGQ}}=\Ordo{\text{TOL}^{-1}}$. This complexity is significantly better than  $\Ordo{\text{TOL}^{-3}}$ achieved by the MC method.
\begin{remark}[On the error  of the  QMC  method combined with numerical smoothing]\label{rem: QMC_smoothing_error}
	Let  $Q^{\text{rQMC}}$  denote the rQMC estimator  to approximate $\expt{\bar{I}}$ in  \eqref{eq: pre_integration_step_wrt_y1_basket}  with $M_{\text{rQMC}}$ samples. Then, we achieve an error decomposition similar to that  in \eqref{eq: error decomposition}, with  \text{Error III} being  the rQMC  statistical error in this case \cite{niederreiter1992random}, and expressed as follows:
	\begin{equation*}\label{eq:Error III, QMC, order}
		\text{Error III} \: \text{(rQMC  error)}=\Ordo{M_{\text{rQMC}}^{-\frac{1}{2}-\delta}  \left(\log M_{\text{rQMC}}\right)^{d \times N-1}},
	\end{equation*}
where  $0 \le \delta \le \frac{1}{2}$  is related to the degree of regularity of  $\bar{I}$, defined in \eqref{eq:expression_h_bar}.

Moreover, the  analysis of QMC with randomly shifted lattice rules  in  \cite{sloan1998quasi,dick2013high} indicates that convergence rates close to the optimal rates $\Ordo{M_{\text{rQMC}}^{-1}}$ can be observed if our integrand, $\bar{I}(\cdot)$, belongs to  the $(dN-1)$-dimensional weighted Sobolev space of functions with square-integrable mixed first derivatives, $\mathcal{W}_{dN-1,\boldsymbol{\gamma}}$,\footnote{$\mathcal{W}_{dN-1,\boldsymbol{\gamma}}$ is  equipped with the (unanchored) norm $
				|| f||^2_{\mathcal{W}_{dN-1,\boldsymbol{\gamma}}} = \sum_{ \boldsymbol{\alpha} \subseteq  \{1:dN-1\}} \frac{1}{\gamma_\alpha} \int_{[0,1]^{| \boldsymbol{\alpha} |}}	 \left(  \int_{[0,1]^{dN-1-| \boldsymbol{\alpha}|}}  \frac{\partial^{| \boldsymbol{\alpha}|}}{\partial \mathbf{y}_{ \boldsymbol{\alpha}} } f(\mathbf{y}) d\mathbf{y}_{- \boldsymbol{\alpha}}\right)^2 d \mathbf{y}_{ \boldsymbol{\alpha}}$,
			where $\mathbf{y}:=(y_j)_{j \in  \boldsymbol{\alpha}}$ and $\mathbf{y}_{- \boldsymbol{\alpha}}:=(y_j)_{j \in \{1:dN-1\} \setminus  \boldsymbol{\alpha}}$.}  where $\boldsymbol{\gamma}:= \{\gamma_\alpha > 0: \boldsymbol{\alpha} \subseteq  \{1, 2, \dots, dN-1\}\}$ is a given collection of weights.
\end{remark}
\begin{remark}
Although  we did not use the Richardson extrapolation in the previous analysis, this hierarchical representation improves the complexity rate of the ASGQ method (as  observed in our numerical experiments in Section \ref{sec:Numerical experiments: Numerical smoothing with ASGQ}).
\end{remark}
\begin{remark}
As an alternative  method to approximate \eqref{eq: option price as integral_basket}, we can use the multilevel MC (MLMC) method, which also benefits  from  the numerical smoothing  (see \cite{bayer2020multilevel}) in terms of complexity and robustness, where we recover complexities obtained for Lipschitz functionals, $\mathcal{O}\left(\text{TOL}^{-2} (\log(\text{TOL}))^2\right)$ when using the Euler–Maruyama  scheme with numerical smoothing .  The comparison OF MLMC and deterministic quadrature methods, such as ASGQ,  is not  straightforward and is problem-dependent because there is a compromise between the regularity class  of the integrand and the anisotropy w.r.t.~the different dimensions. We intend to conduct this systematic comparison in future work. We also plan to explore the idea of numerical smoothing with multilevel QMC \cite{giles2009multilevel}, where we can profit from the good features of QMC and MLMC in this setting.
\end{remark}

\section{Numerical Experiments}\label{sec:Numerical experiments: Numerical smoothing with ASGQ}
We conduct  experiments using three  examples of payoffs: a single-asset digital option,  a single-asset call option, and  a 	four-asset arithmetic basket call option.\footnote{The payoff  $g$ is  expressed by $	g(\mathbf{x})=\max\left(\sum_{j=1}^{d} c_{j} x^{(j)}-K,0  \right)$, where $\{c_j\}_{j=1}^d$ denote the weights of the basket.} These examples are tested under two  dynamics for the asset price:   the discretized GBM model (a didactic example) and  the   Heston model, which is a relevant application of our approach (discretization is required).   Table \ref{table: Examples details.} lists the specifications of each example. Further details of the models and discretization schemes are described  in Section \ref{sec: Experiments settings}.  Sections \ref{sec: ASGQ with numerical smoothing versus ASGQ without  smoothing} and \ref{sec: QMC with numerical smoothing versus QMC without  smoothing}  demonstrate the advantage of combining numerical smoothing with the ASGQ  and rQMC methods over the  ASGQ and rQMC without  smoothing. In Section \ref{sec: Study of the numerical smoothing parameters},   we study the effect of the  numerical smoothing parameters on the numerical smoothing error and consequently on the quadrature error of the ASGQ method. Finally,   Section  \ref{sec: ASGQ with numerical smoothing  versus MC without smoothing} compares  the  MC and ASGQ methods in terms of errors and computational times. Our ASGQ implementation was based on \url{https://sites.google.com/view/sparse-grids-kit}.
\begin{small}
	\begin{table}[h!]
		\centering
		\begin{small}
			\begin{tabular}{||l|*{3}{|c|}r}
				\toprule[1.5pt]
			Example &	Parameters            & Reference solution    \\
				\hline
				
				Single-asset digital option under  GBM  & $\sigma=0.4$, $r=0$,	$T=1$,   $S_0=K=100$  & $0.42074$  \\	
					\hline
				Single-asset digital option under  Heston  &	 $v_0=0.04$, $\mu=0$,  $\rho=-0.9$, $\kappa=1$, $\xi=0.1$, & $\underset{(2.0e-05)}{0.5146}$  \\
				 &	 $\theta=0.0025$, $S_0=K=100$&   \\
					\hline
				Single-asset call option under  GBM  & 	$\sigma=0.4$, $r=0$,	$T=1$,   $S_0=K=100$   & $15.8519$  \\
					\hline
					Single-asset call  option under  Heston  &	$v_0=0.04$, $\mu=0$,  $\rho=-0.9$, $\kappa=1$, $\xi=0.1$, & $6.33254$  \\
					&	 $\theta=0.0025$, $S_0=K=100$ &   \\
						\hline
					
					$4$-asset basket call option   &$\sigma_{1,2,3,4}=0.4$, $\rho=0.3$, $r=0$, $T=1$,    & $\underset{(1.0e-03)}{11.04}$    \\
					under  GBM	& $S_0^{1,2,3,4}=K=100$,   $c_{1,2,3,4}=1/4$  &   \\
				\bottomrule[1.25pt]
			\end{tabular}
		\end{small}
		\caption{Model and option parameters of the tested examples  with their reference solutions.   The reference solution for the call option under the Heston model is computed  using Premia software using the method  in \cite{heston1993closed}. The numbers between parentheses correspond to the statistical error estimates when the reference solution is estimated using the MC estimator.} 
		\label{table: Examples details.}
	\end{table}
\end{small}
\subsection{Experiments setting}\label{sec: Experiments settings}
Regarding  the  numerical experiments  under the  GBM model, the assets dynamics follow \eqref{eq:lognormal model} and  are simulated using the Euler–Maruyama  scheme. Moreover, we test options under the Heston model \cite{heston1993closed,broadie2006exact,kahl2006fast,andersen2007efficient}, providing the following dynamics:
	\begin{align}\label{eq:dynamics Heston}
		dS_t&=\mu S_t dt+\sqrt{v_t}S_t dW_t^S= \mu S_t dt+ \rho\sqrt{v_t}S_t dW_t^v+ \sqrt{1-\rho^2} \sqrt{v_t}S_t dW_t \nonumber\\
		dv_t&=\kappa (\theta-v_t)dt+\xi \sqrt{v_t} dW_t^v\COMMA
\end{align}
where  $S_t$ denotes the asset price, $v_t$ represents the instantaneous variance, $\left(W_{t}^{S},W_{t}^{v}\right)$ are the correlated Wiener processes with correlation $\rho$, $\mu$  represents the asset's rate of return, $\theta$ is  the mean  variance, $\kappa$ is the rate at which $v_t$ reverts to $\theta$, and $\xi$ denotes the volatility of the volatility.

Many simulation schemes of \eqref{eq:dynamics Heston} have been proposed in the literature.    These methods primarily differ in how they simulate the volatility process to ensure positivity. Appendix \ref{sec:Schemes to simulate the Heston dynamics} provides an overview of the most popular methods in this context.

The ASGQ and rQMC  methods are extremely sensitive to the smoothness of the integrand. In particular, we  numerically found (Appendix  \ref{sec:On the choice of the simulation scheme of the Heston model}) that using  a nonsmooth transformation to ensure the positivity of the volatility process deteriorates the performance of the ASGQ method. To overcome this undesirable feature,   we propose using  an alternative scheme, namely, the Heston  Ornstein–Uhlenbeck (OU)-based scheme, in which the volatility is simulated as the sum of the OU or Bessel processes (Appendix \ref{sec:Discretization of Heston model with the volatility process Simulated using the sum of  Ornstein-Uhlenbeck (Bessel) processes}).  In the literature  \cite{andersen2007efficient, lord2010comparison,alfonsi2010high}, the focus has been on designing schemes that ensure the positivity of the volatility process and exhibit a good weak error behavior. In our setting, an optimal scheme is determined based on two criteria: (i) the behavior of the rates of mixed differences, which is an important feature for ensuring the    optimal performance of the ASGQ method (see  Appendix \ref{sec:Comparison in terms of  the weak error behavior} for more details), and  (ii) the weak error behavior  to apply the Richardson extrapolation when necessary. Comparing the different schemes (see Appendices \ref{sec:Comparison in terms of mixed differences rates} and \ref{sec:Comparison in terms of  the weak error behavior}) suggests that  the Heston OU-based scheme  yields the best results based on  our criteria.   Therefore, we used this scheme with the ASGQ and rQMC methods in our numerical experiments. For the MC method, we used the full truncation scheme (explained in Appendix \ref{sec:Discretization of Heston model with a non smooth transformations for the volatility process}).
\begin{remark}
In this work, our primary focus is the numerical smoothing idea with the  implied additional regularity  for the quantity of interest and its benefits on the performance of ASGQ and QMC. As a byproduct, for the examples under the Heston model, we numerically found that using a nonsmooth transformation to ensure the positivity of the volatility process deteriorates the   performance of ASGQ  even after applying the numerical smoothing because it affects the path regularity of the process. To overcome this undesirable feature for the parameters settings that we consider ($4 \kappa \theta/\xi^2  $ is an integer), we suggest using the Heston OU-based scheme as an alternative.   We expect our observations to still be valid for cases with tiny perturbations of the Heston model parameters ($4 \kappa \theta/\xi^2  $ is very close to an integer).  An extensive analysis of the proposed  Heston OU-based scheme is left for future work, where we plan to conduct a systematic investigation of its performance compared to the popular existing schemes, in the same spirit as  \cite{lord2010comparison}, and examine various challenging settings of model parameters (e.g.,  when $\xi^2 \gg 4 \kappa \theta $   and  $4 \kappa \theta/\xi^2  $    is not an integer).  The above observations suggest that, besides smoothing out the observable,  the regularity of the discretization scheme is also essential. Numerical smoothing   works perfectly if the scheme has sufficient path regularity. Otherwise,  besides  numerical smoothing,  a smooth discretization scheme must preserve the process path regularity to ensure the optimal performance of ASGQ and QMC.
	\end{remark}
\subsection{Comparison of the ASGQ method with and without numerical smoothing}\label{sec: ASGQ with numerical smoothing versus ASGQ without  smoothing}
This section illustrates  the advantage of combining numerical smoothing with the ASGQ method. Figures  \ref{fig: Digital option under Heston: Comparing the relative quadrature error convergence for ASGQ combined with numerical smoothing and ASGQ without smoothing.} and \ref{fig: Call option under Heston: Comparing the relative quadrature error convergence for ASGQ combined with numerical smoothing and ASGQ without smoothing.} show comparisons of 
	the relative quadrature error convergence for the  examples under the Heston model  in Table \ref{table: Examples details.},  with and without the Richardson  extrapolation.\footnote{The dimension of the integration problem is $N$ for the GBM examples and $2N$ for the Heston examples.}  Numerical smoothing significantly improves  the quadrature error convergence for all cases, which agrees with Theorem \ref{thr:smoothness}. For instance, for the call option under the Heston model (left plot in Figure \ref{fig: Call option under Heston: Comparing the relative quadrature error convergence for ASGQ combined with numerical smoothing and ASGQ without smoothing.}), the ASGQ method without smoothing cannot achieve  a relative quadrature error below $10\%$, even in the case of   more than $10^3$ quadrature points. Alternatively, the ASGQ method with numerical smoothing achieves a relative quadrature error below $1\%$ with the same number of quadrature points. The gains are more evident in the digital option case than in the call option case  (Figure \ref{fig: Digital option under Heston: Comparing the relative quadrature error convergence for ASGQ combined with numerical smoothing and ASGQ without smoothing.}). Further, using the Richardson extrapolation, the ASGQ method with numerical smoothing yields a smaller quadrature error. 	For all cases of options and  models,  with or without Richardson extrapolation, we observe that numerical smoothing also reduces the constant in the quadrature error besides improving the convergence rate. This observation can be explained by the analysis in \cite{chen2018sparse}, indicating that the constant in the error estimate depends on the  weighted sum of the mixed derivatives of the integrand. From this perspective, the numerical smoothing enables a faster decay of mixed derivatives than the case without smoothing (see Proposition 3.4 in \cite{chen2018sparse}). 
%
\begin{figure}[h!]
	\centering 
	\begin{subfigure}{0.4\textwidth}
		\includegraphics[width=\linewidth]{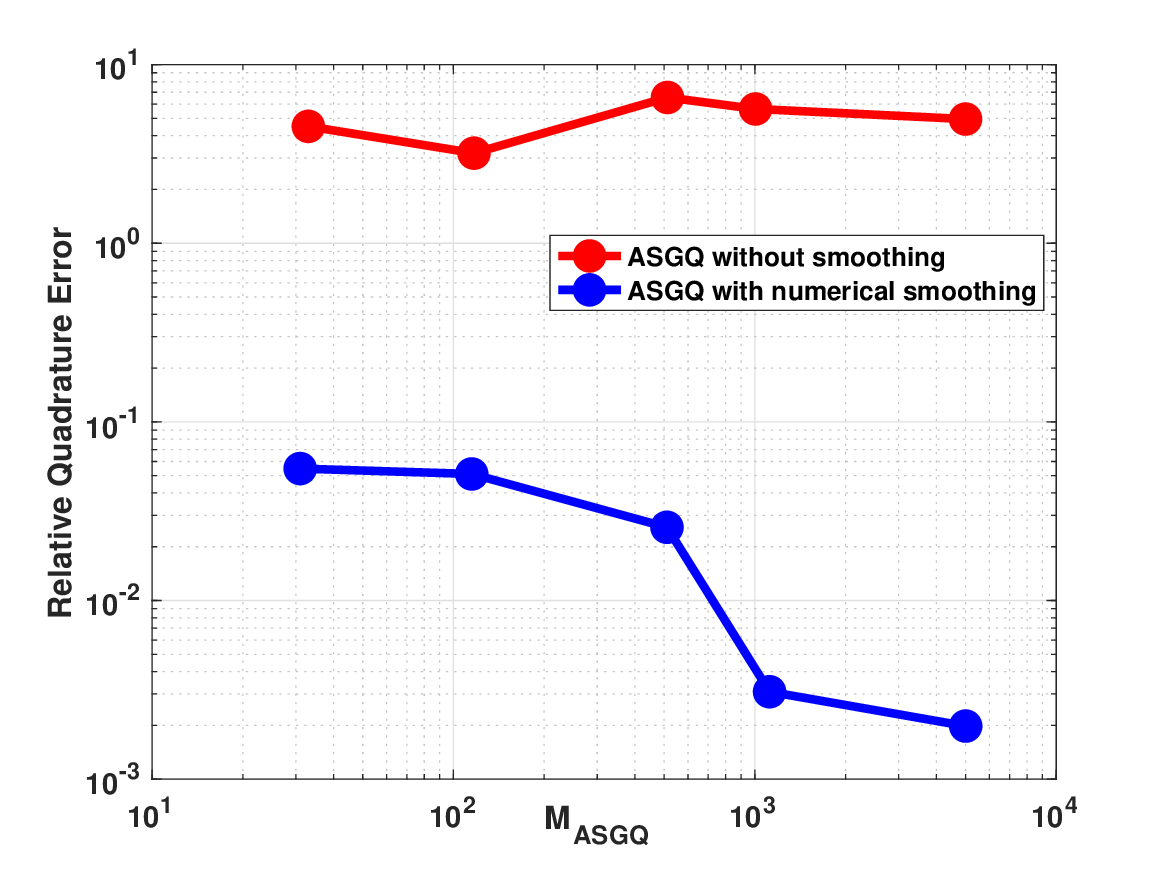}
		\caption{}
		\label{fig:1}
	\end{subfigure}\hfil 
	\begin{subfigure}{0.4\textwidth}

			\includegraphics[width=\linewidth]{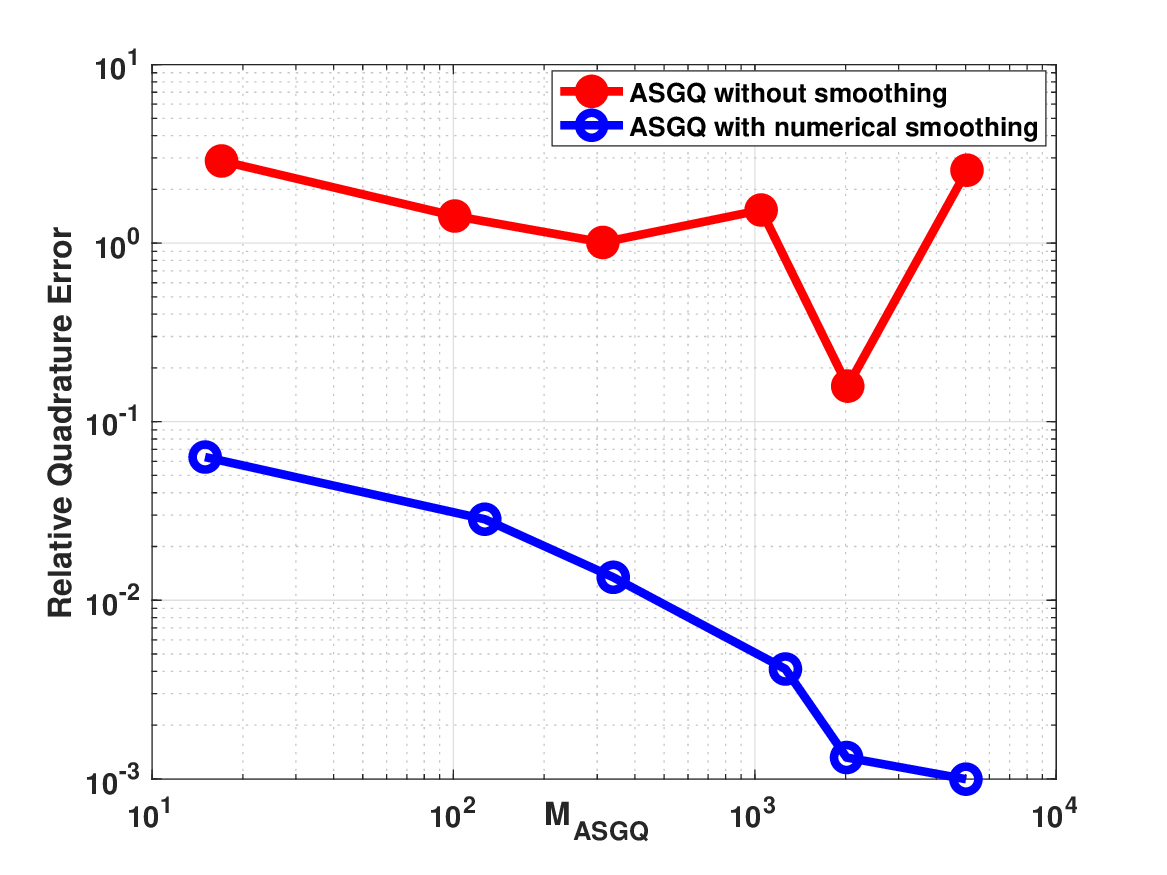}
		\caption{}
		\label{fig:2}
	\end{subfigure}
	\caption{Digital option under  the Heston model: Comparison of the relative quadrature error convergence for the ASGQ method with and without numerical smoothing.  (a) Without the Richardson extrapolation ($N=8$), and (b) with the Richardson extrapolation ($N_{\text{fine level}}=8$).}
	\label{fig: Digital option under Heston: Comparing the relative quadrature error convergence for ASGQ combined with numerical smoothing and ASGQ without smoothing.}
\end{figure}
\begin{figure}[h!]
	\centering 
	\begin{subfigure}{0.4\textwidth}
		\includegraphics[width=\linewidth]{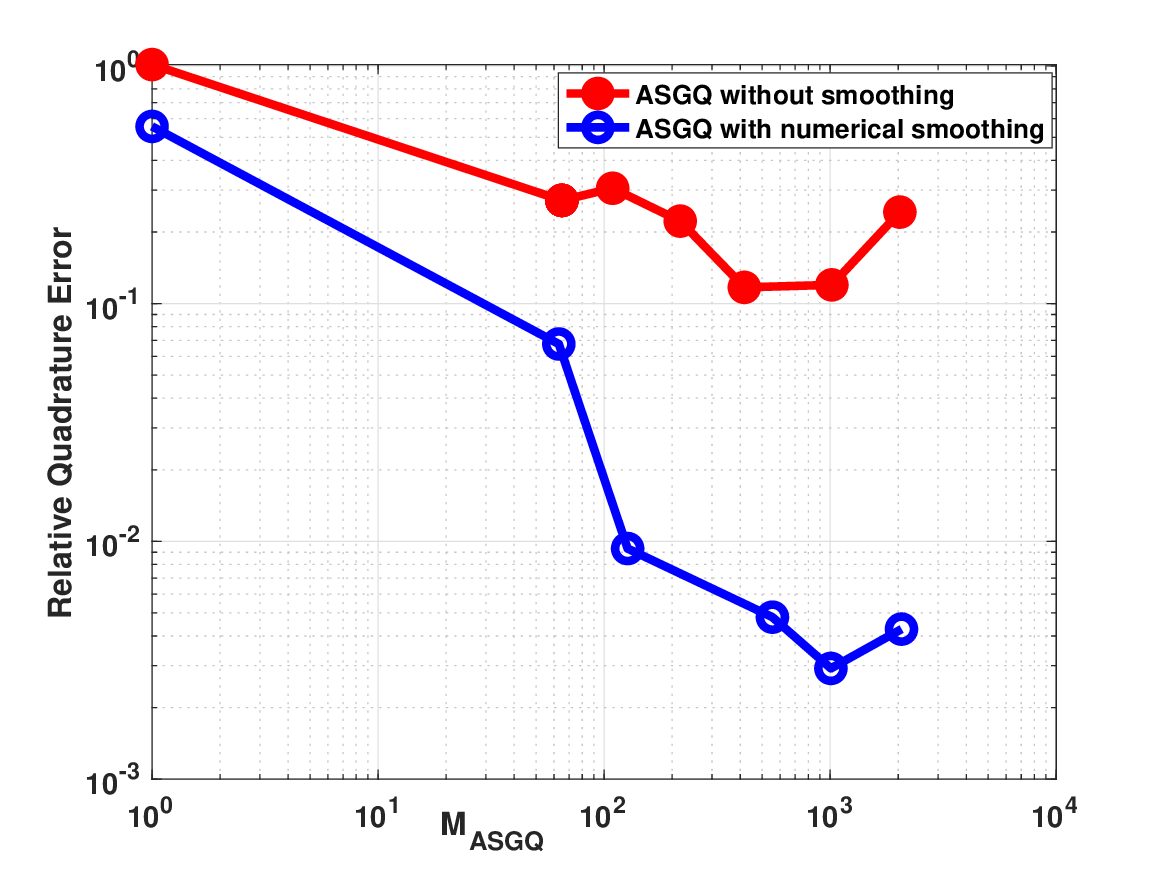}
		\caption{}
		\label{fig:1}
	\end{subfigure}\hfil 
	\begin{subfigure}{0.4\textwidth}
	\includegraphics[width=\linewidth]{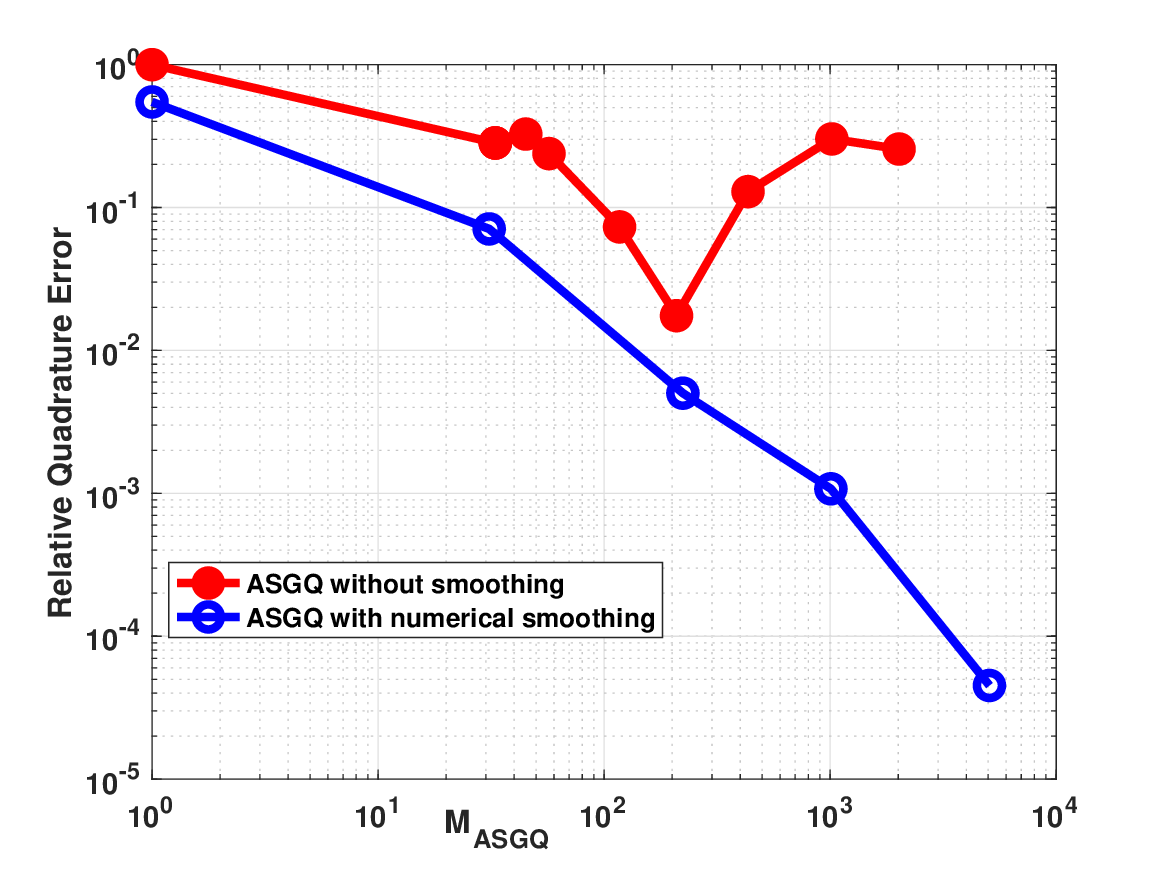}
		\caption{}
		\label{fig:2}
	\end{subfigure}
	\caption{Call option under the  Heston model:  Comparison of the relative quadrature error convergence for the ASGQ method with and without numerical smoothing.  (a) Without the Richardson extrapolation ($N=16$), and (b) with the Richardson extrapolation ($N_{\text{fine level}}=8$).}
	\label{fig: Call option under Heston: Comparing the relative quadrature error convergence for ASGQ combined with numerical smoothing and ASGQ without smoothing.}
\end{figure}
\subsection{Comparison of the rQMC method with and without numerical smoothing}\label{sec: QMC with numerical smoothing versus QMC without  smoothing}
In this section, we demonstrate the advantage of combining numerical smoothing with  the rQMC method. Figures  \ref{fig: Comparing the  statisitical error convergence for QMC combined with numerical smoothing versus QMC without smoothing. a) Digital option under GBM.  b) Call option under GBM..}  and  \ref{fig: Comparing the statisitical error convergence for QMC combined with numerical smoothing versus QMC without smoothing. a) Digital option under Heston (N=4).  b) Call option under Heston $(N=4).}  display comparisons of 
the  statistical error convergence for the examples   in Table \ref{table: Examples details.} and for some  number of time steps $N$.  Because regularity was regained using numerical smoothing, an improvement in the statistical error convergence of the rQMC method occurs, which agrees with Remark \ref{rem: QMC_smoothing_error} (see left plots in Figures \ref{fig: Comparing the  statisitical error convergence for QMC combined with numerical smoothing versus QMC without smoothing. a) Digital option under GBM.  b) Call option under GBM..}  and  \ref{fig: Comparing the statisitical error convergence for QMC combined with numerical smoothing versus QMC without smoothing. a) Digital option under Heston (N=4).  b) Call option under Heston $(N=4).}).
For  the Heston and GBM models,   the convergence rate for QMC was  improved more significantly   for the digital option than the call option payoff. 
\begin{figure}[h!]
	\centering 
	\begin{subfigure}{0.4\textwidth}
		\includegraphics[width=\linewidth]{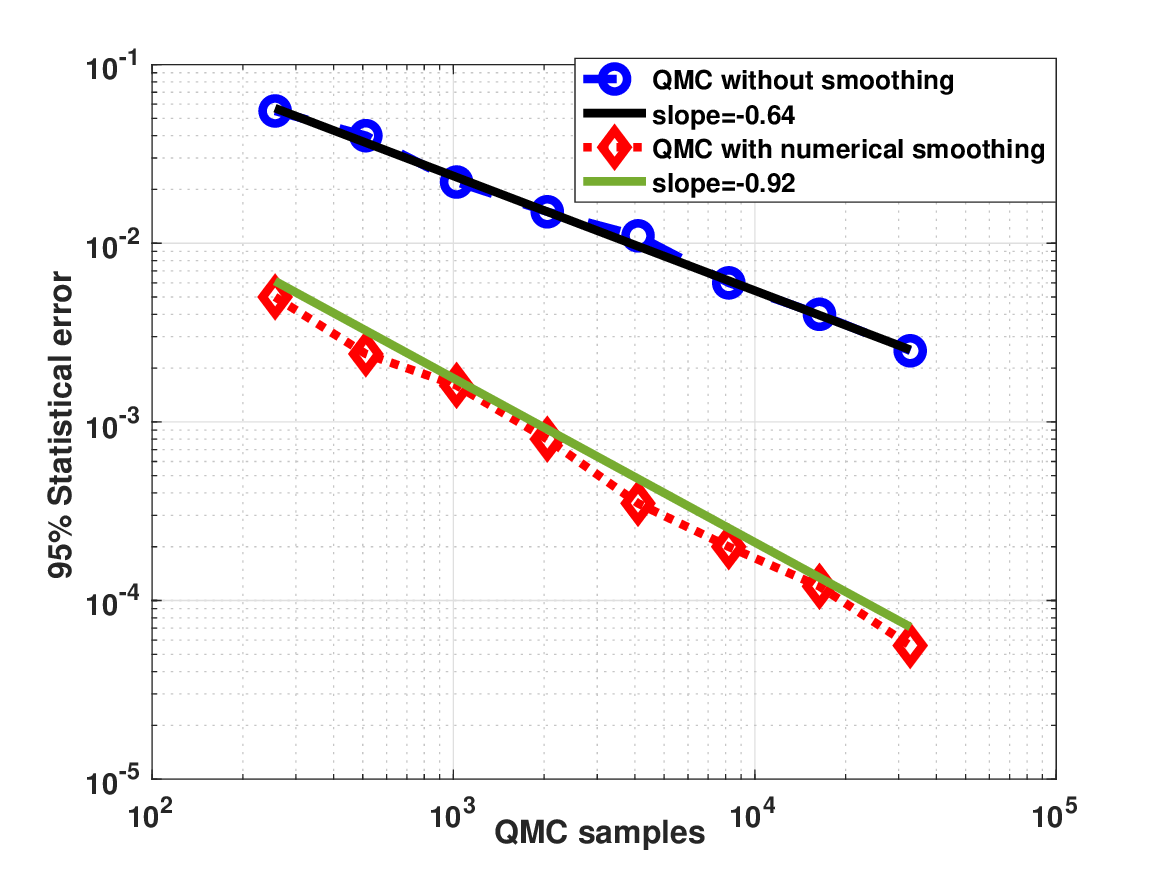}
		\caption{}
		\label{fig:1}
	\end{subfigure}\hfil 
	\begin{subfigure}{0.4\textwidth}
		\includegraphics[width=\linewidth]{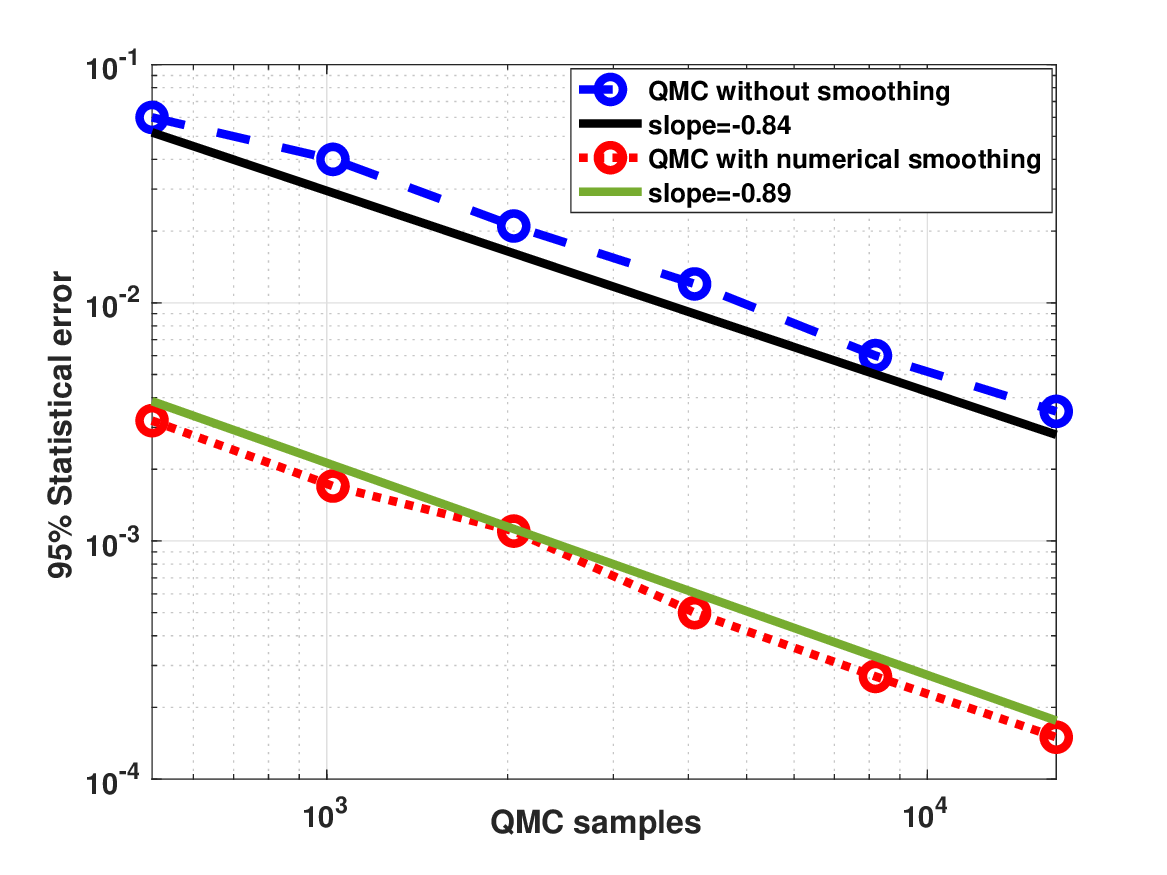}
		\caption{}
		\label{fig:2}
	\end{subfigure}
	\caption{Comparison of the $95\%$ statistical error convergence for  rQMC  with and without numerical smoothing with $N=8$. (a) Digital option under GBM,    (b) call option under GBM.}
	\label{fig: Comparing the  statisitical error convergence for QMC combined with numerical smoothing versus QMC without smoothing. a) Digital option under GBM.  b) Call option under GBM..}
\end{figure}
\begin{figure}[h!]
	\centering 
	\begin{subfigure}{0.4\textwidth}
		\includegraphics[width=\linewidth]{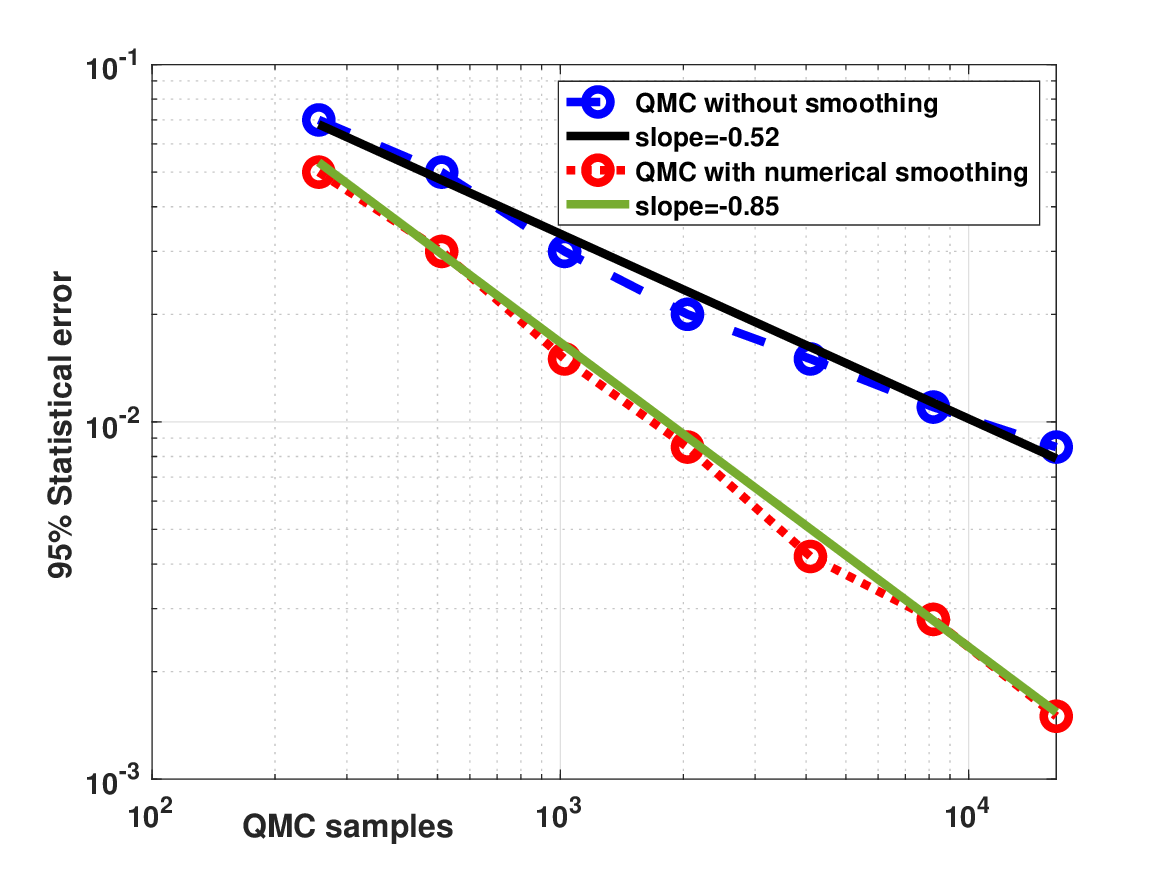}
		\caption{}
		\label{fig:1}
	\end{subfigure}\hfil 
	\begin{subfigure}{0.4\textwidth}
		\includegraphics[width=\linewidth]{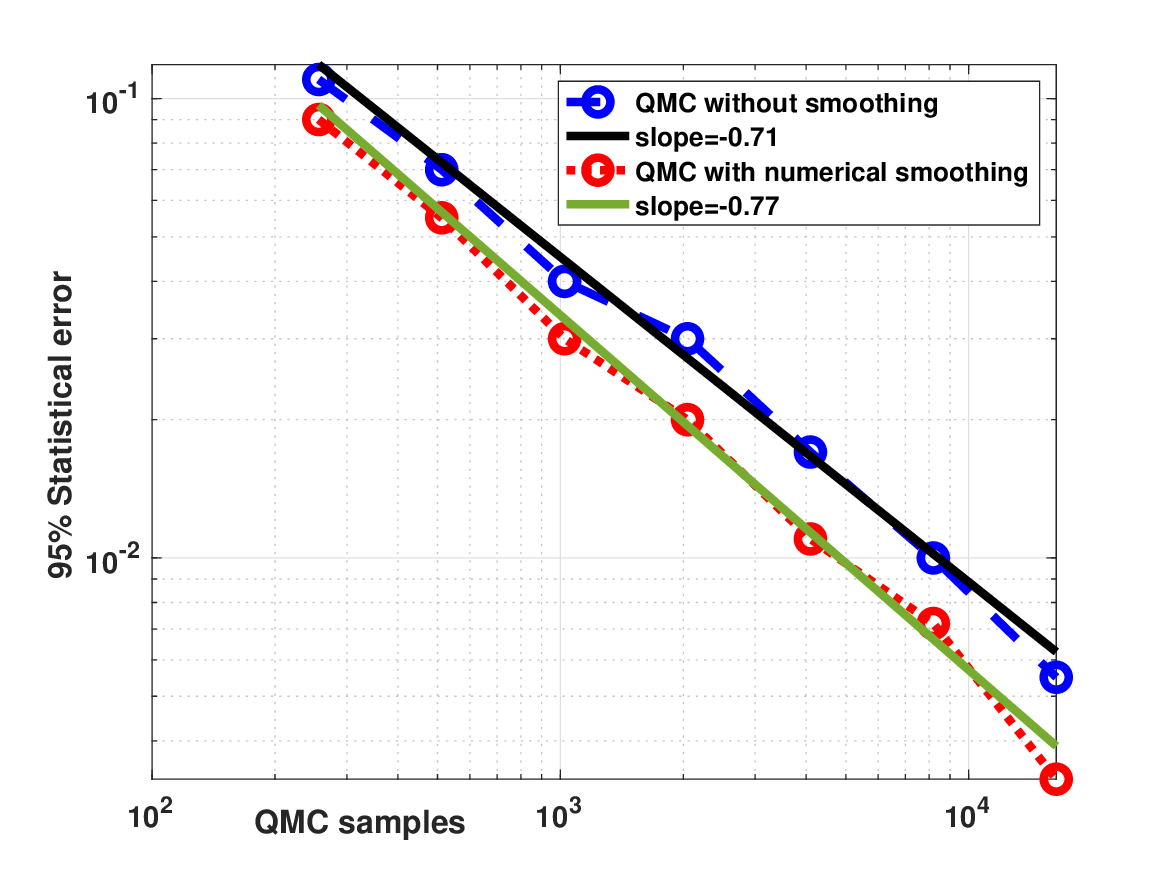}
		\caption{}
		\label{fig:2}
	\end{subfigure}
	\caption{Comparison of the $95\%$ statistical error convergence for  rQMC  with and without numerical smoothing  with $N=4$. (a) Digital option under Heston,     (b) call option under Heston.}
	\label{fig: Comparing the statisitical error convergence for QMC combined with numerical smoothing versus QMC without smoothing. a) Digital option under Heston (N=4).  b) Call option under Heston $(N=4).}
\end{figure}
\subsection{Study of the numerical smoothing parameters}\label{sec: Study of the numerical smoothing parameters}
We study the effect of the   smoothing parameters on the relative numerical smoothing error for  sufficiently large ASGQ points $M_{\text{ASGQ}}=10^3$.  These parameters are (i) the number of Laguerre points  in the preintegration step, $M_{\text{Lag}}$, and (ii) the Newton tolerance  in the root-finding step, $\text{TOL}_{\text{Newton}}$. Figures \ref{fig: Digital option under GBM with N=4: The relative quadrature error convergence for ASGQ combined with numerical smoothing for different scenarios of smoothing parameters. } and  \ref{fig: Call option under GBM with N=4: The relative quadrature error convergence for ASGQ combined with numerical smoothing for different scenarios of smoothing parameters. } present the digital and call option results under the GBM model in Table \ref{table: Examples details.}, for the case without the Richardson extrapolation and when $N=4$.  These plots show that a faster convergence of the root-finding and quadrature errors can be achieved for the call option compared to the digital option. Moreover, we observe that  the numerical smoothing procedure is  cheap because  few Laguerre quadrature points and  large values of Newton tolerance are required to achieve a certain accuracy. Similar observations have been obtained for other  examples.
\begin{figure}[h!]
	\centering 
	\begin{subfigure}{0.4\textwidth}
		\includegraphics[width=\linewidth]{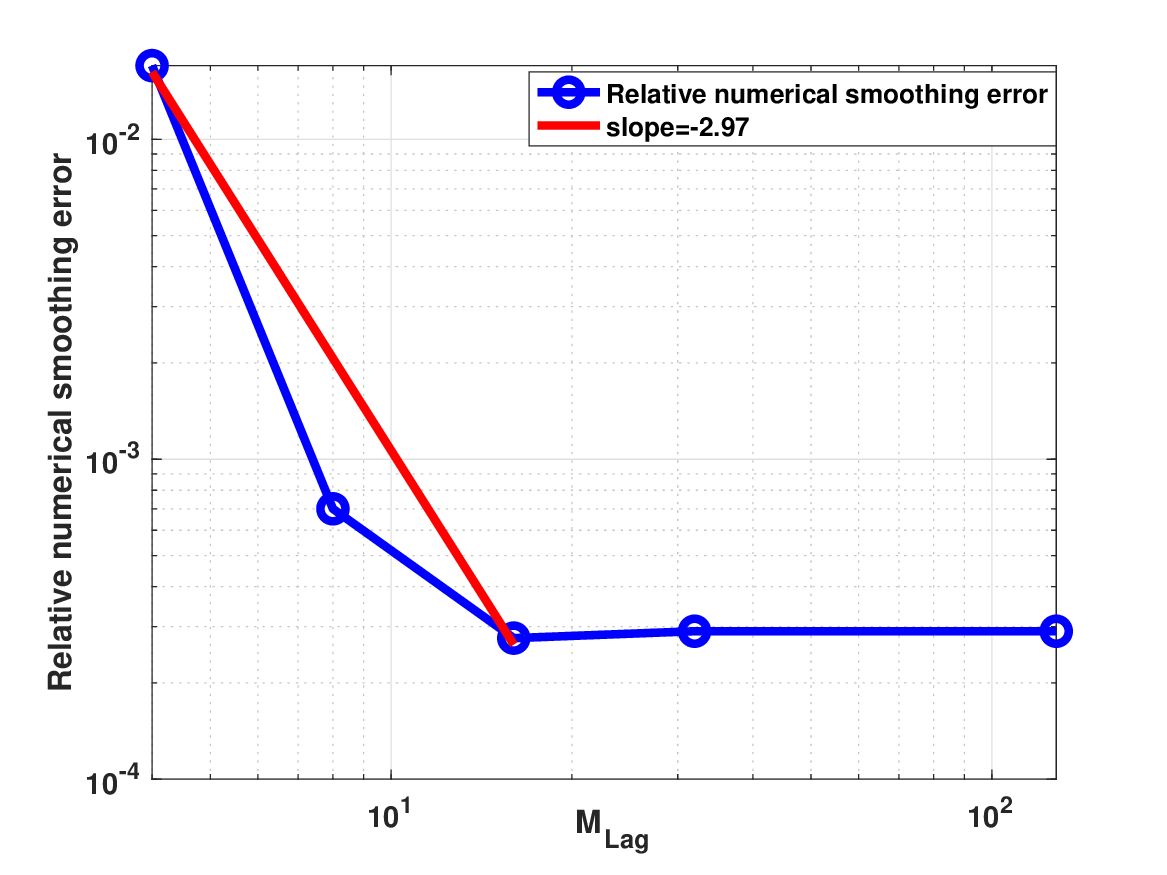}
		\caption{}
		\label{fig:1}
	\end{subfigure}\hfil 
	\begin{subfigure}{0.4\textwidth}
		\includegraphics[width=\linewidth]{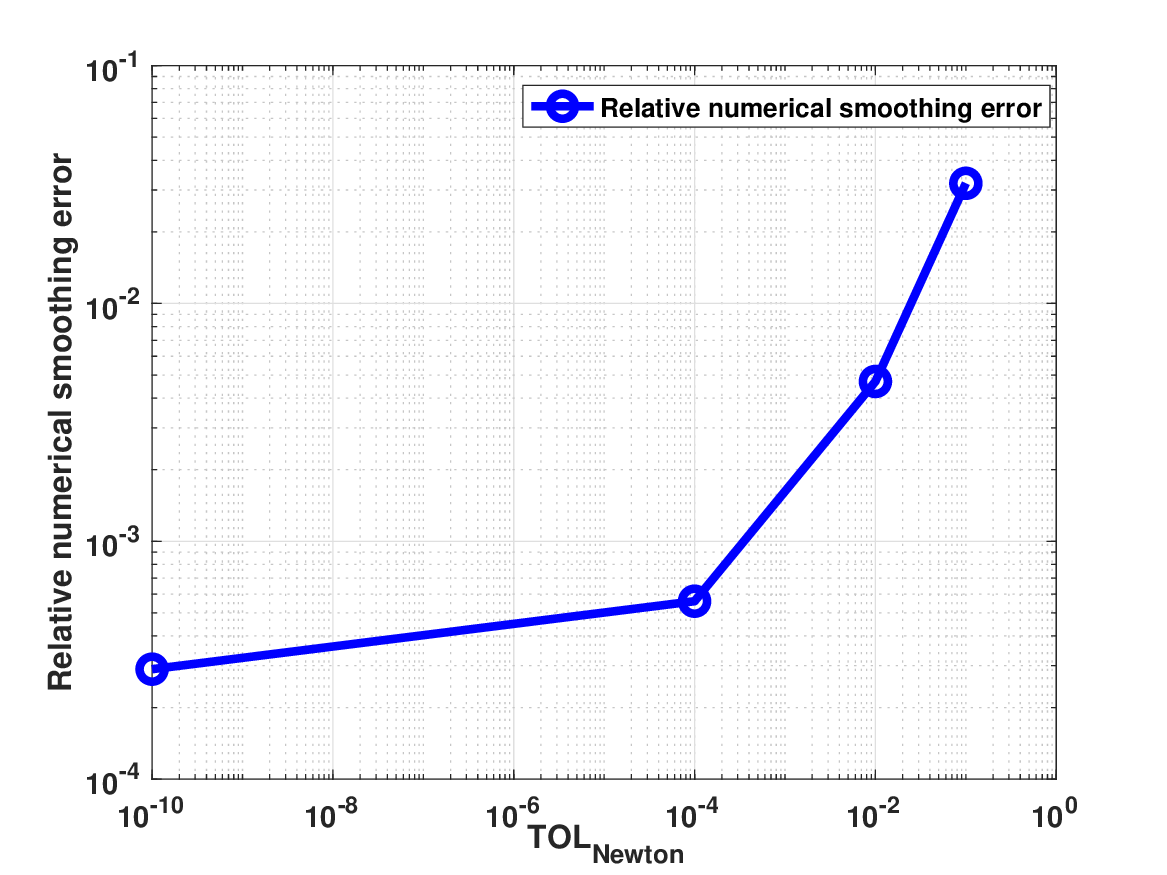}
		\caption{}
		\label{fig:2}
	\end{subfigure}
	\caption{Digital option under GBM with $N=4$: The relative numerical smoothing error for a fixed number of ASGQ points $M_{\text{ASGQ}}=10^3$ plotted against (a)  different values of $M_{\text{Lag}}$ with a fixed Newton tolerance $\text{TOL}_{\text{Newton}}=10^{-10}$, (b)   different values of $\text{TOL}_{\text{Newton}}$ with a fixed number of Laguerre quadrature points   $M_{\text{Lag}}=128$.}
	\label{fig: Digital option under GBM with N=4: The relative quadrature error convergence for ASGQ combined with numerical smoothing for different scenarios of smoothing parameters. }
\end{figure}
\begin{figure}[h!]
	\centering 
	\begin{subfigure}{0.4\textwidth}
		\includegraphics[width=\linewidth]{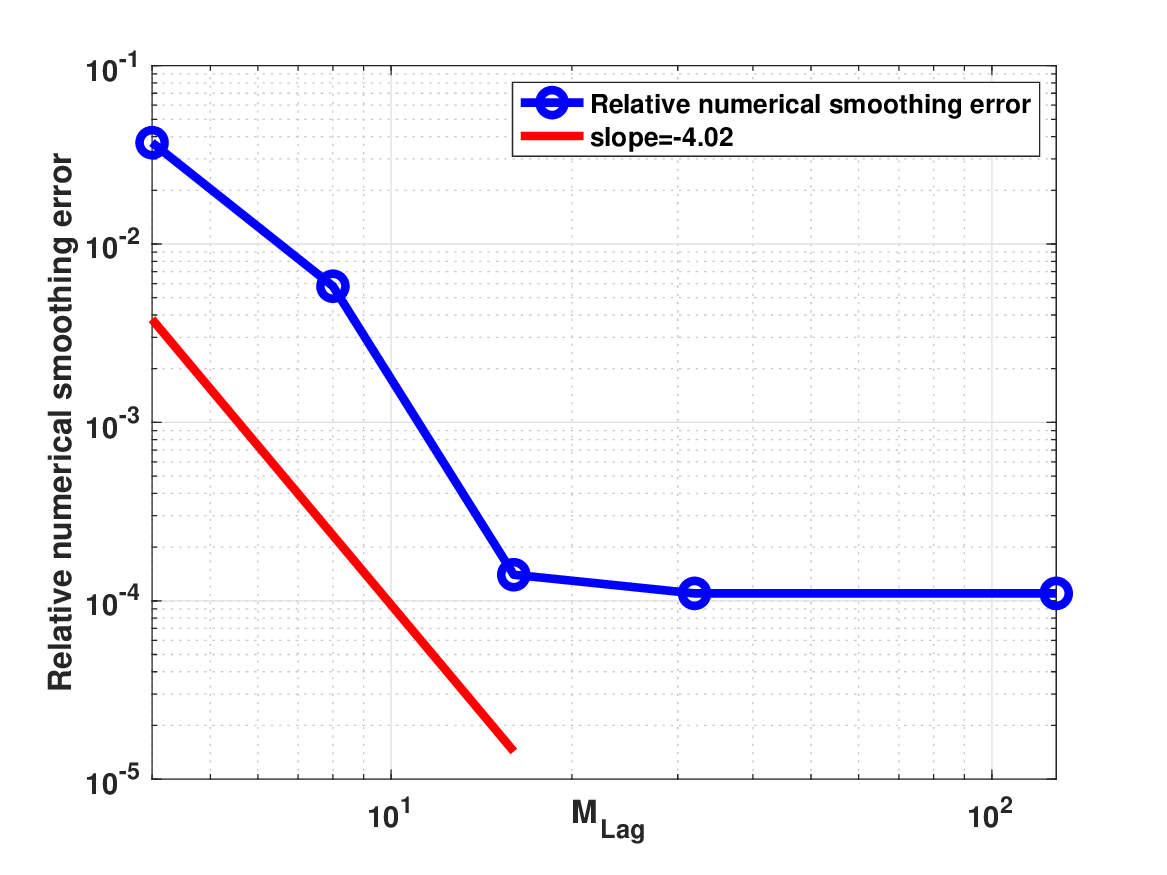}
		\caption{}
		\label{fig:1}
	\end{subfigure}\hfil 
	\begin{subfigure}{0.4\textwidth}
		\includegraphics[width=\linewidth]{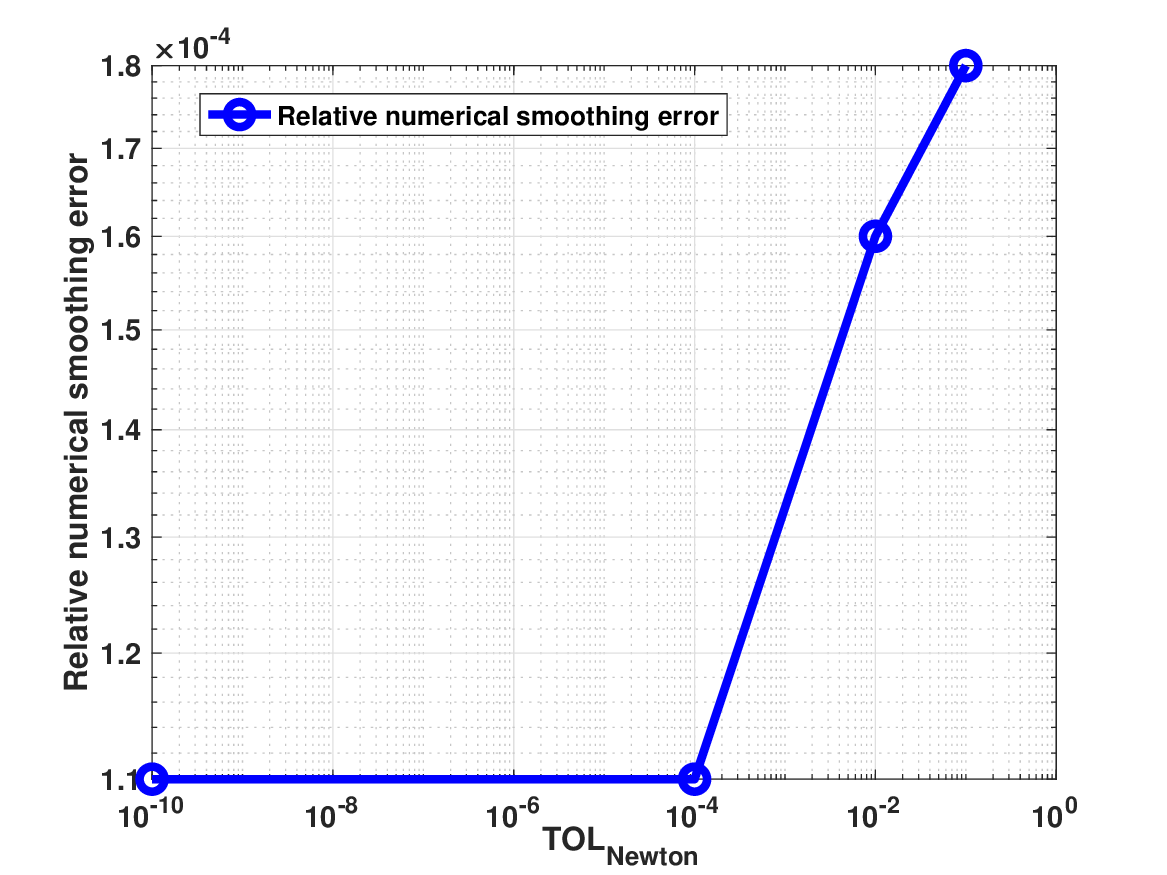}
		\caption{}
		\label{fig:2}
	\end{subfigure}
	\caption{Call option under GBM with  $N=4$: The  relative numerical smoothing error for a fixed number of ASGQ  points $M_{\text{ASGQ}}=10^3$ plotted against (a)  different values of $M_{\text{Lag}}$ with  a  fixed Newton tolerance $\text{TOL}_{\text{Newton}}=10^{-10}$, (b)   different values of $\text{TOL}_{\text{Newton}}$ with a fixed number of Laguerre quadrature points   $M_{\text{Lag}}=128$.}
	\label{fig: Call option under GBM with N=4: The relative quadrature error convergence for ASGQ combined with numerical smoothing for different scenarios of smoothing parameters. }
\end{figure}
\subsection{ASGQ method with numerical smoothing versus the MC method}\label{sec: ASGQ with numerical smoothing  versus MC without smoothing}
For  a sufficiently fixed small error tolerance in the price estimates, we compare the computational time needed for the MC method and the ASGQ method with numerical smoothing  to meet the desired error tolerance. The reported errors are relative errors normalized using the reference solutions.  Furthermore, we conduct our numerical experiments for two  scenarios:  without the Richardson extrapolation, and  with level-$1$   Richardson extrapolation. The actual work (runtime) is obtained using a 3,2 GHz 8-Core Intel Xeon W architecture.

The numerical findings are summarized in Table \ref{table:Summary of our numerical results.}. The reported results highlight the computational gains achieved using the ASGQ method with numerical smoothing compared to the MC method to meet a  relative error   below  $1\%$.  These  results correspond to the best configuration with the Richardson extrapolation for each method. More details  for each case are provided in Figures  \ref{fig: Digital option under GBM and Heston: Computational work comparison for the different methods with the different configurations in terms of the level of Richardson extrapolation. To achieve a relative error below 1,
	the ASGQ method combined with numerical smoothing and level 1 of Richardson extrapolationsignificantly outperforms  the other methods.}, 	\ref{fig: Call option under GBM and Heston: Computational work comparison for the different methods with the different configurations in terms of the level of Richardson extrapolation. To achieve a relative error below 1,		 ASGQ  combined with numerical smoothing and level 1 of Richardson extrapolation
	significantly outperforms  the other methods.} and \ref{fig: 4-dimensional basket call option under GBM: Computational work comparison for the different methods. To achieve a relative error below 1, ASGQ  combined with numerical smoothing significantly outperforms the  MC method.}, comparing the numerical complexity of each method under the two Richardson extrapolation scenarios.  These figures illustrate that,  to achieve a relative error of less than $1\%$, the optimal configuration is the  level-$1$  Richardson extrapolation for both the MC and ASGQ methods, except for the four-asset basket call option under  the  GBM model. 	
\begin{small}
	\begin{table}[!h]
		\centering
		\begin{tabular}{l*{4}{c}r}
			\toprule[1.5pt]
			Example            &  Total relative error  & CPU time  $\left( \text{ASGQ} / \text{MC}\right)$ in \% \\
			\hline
			Single-asset digital option (GBM)  &  $0.4 \%$&  $ 0.2\%$ \\	
			\hline
			Single-asset call option (GBM)      &  $0.5\%$&  $0.3\%$ \\
			
			\hline
			Single-asset  digital option (Heston)   &$0.4\%$&  $3.2\%$ \\	
			\hline
			Single-asset call option (Heston)   &  $0.5\%$&  $0.4\%$ \\
			\hline
				$4$-asset basket   call option (GBM)     &  $0.8\%$&  $7.4\%$ \\	
			\bottomrule[1.25pt]
		\end{tabular}
		\caption{Summary of the relative errors and computational gains achieved using  ASGQ  with numerical smoothing compared to the MC method, to realize a certain error tolerance.  The CPU time ratios are computed for the best configuration with  Richardson extrapolation for each method.}
		\label{table:Summary of our numerical results.}
	\end{table}
\end{small}
\begin{figure}[h!]
	\centering 
	\begin{subfigure}{0.4\textwidth}
		\includegraphics[width=\linewidth]{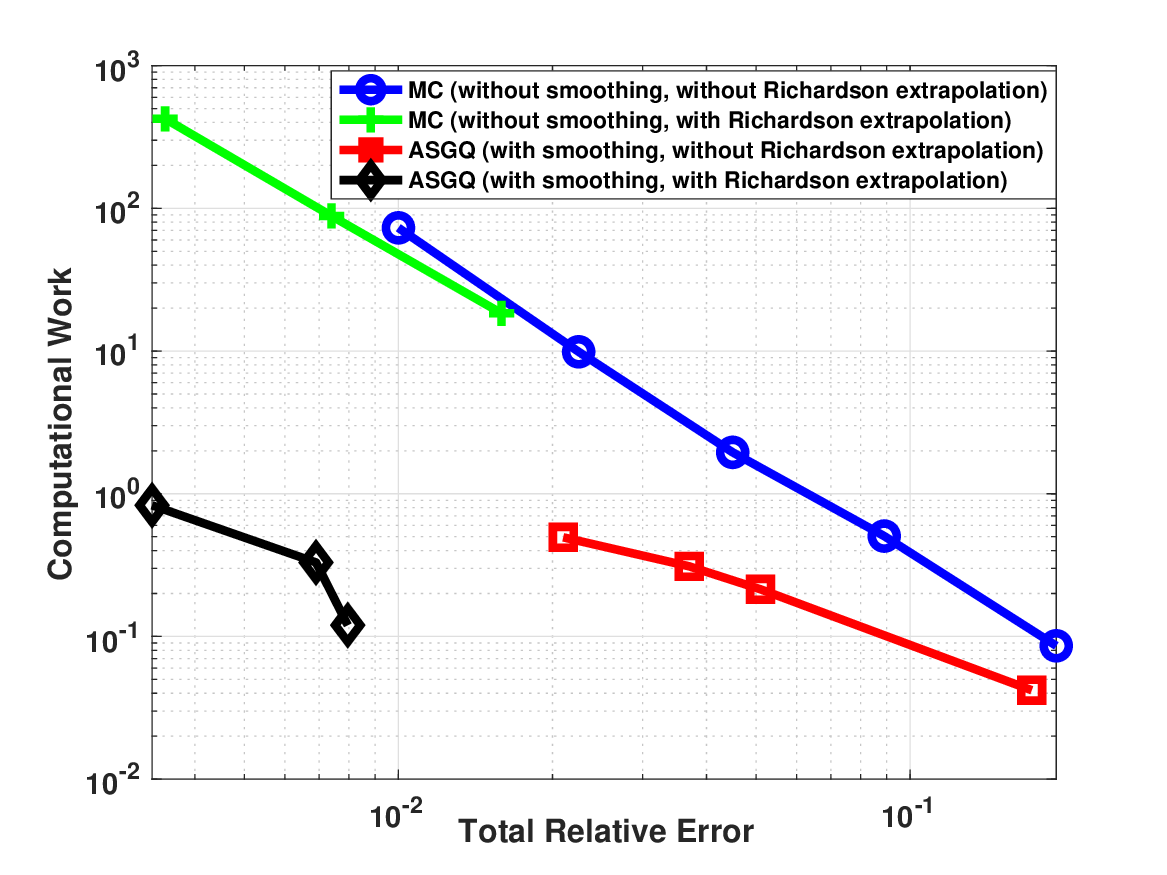}
		\caption{}
		\label{fig:1}
	\end{subfigure}\hfil 
	\begin{subfigure}{0.4\textwidth}
		\includegraphics[width=\linewidth]{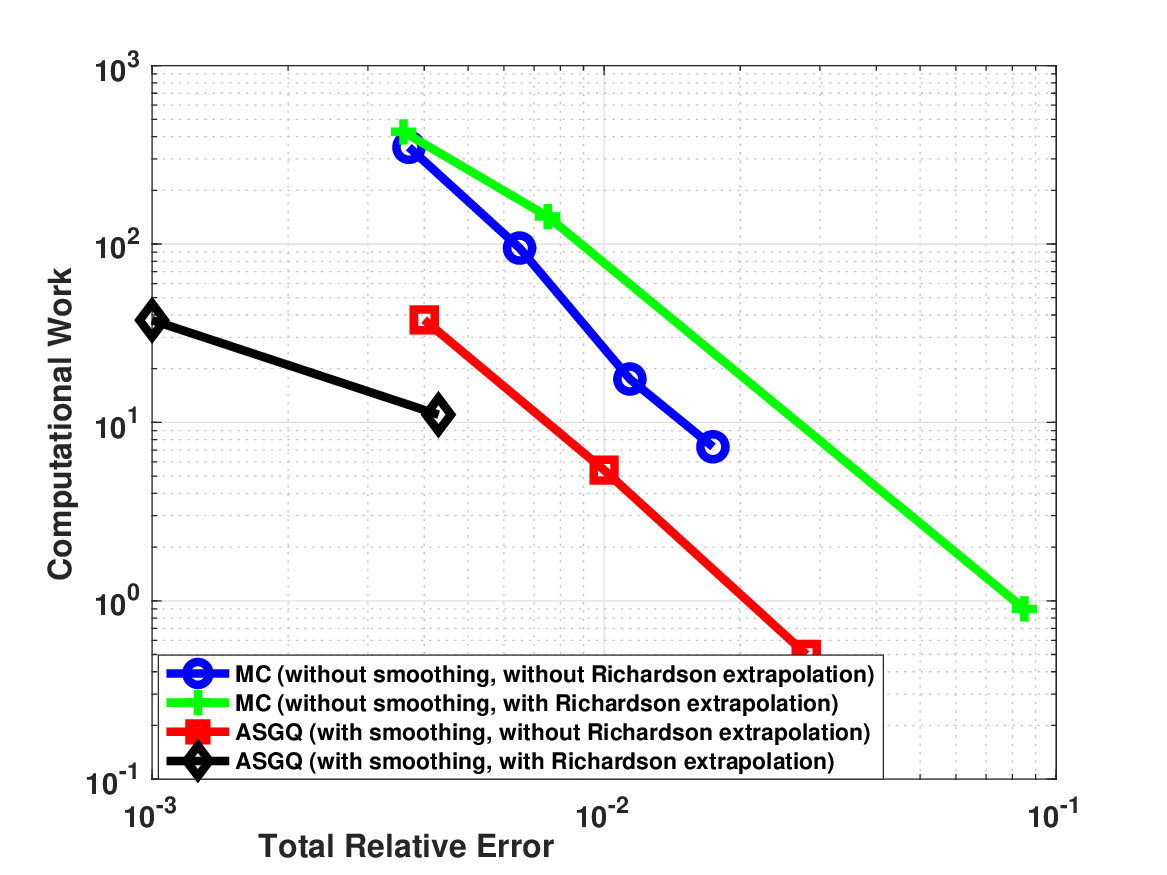}
		\caption{}
		\label{fig:2}
	\end{subfigure}
	\caption{Computational work comparison for the different methods with the various configurations in terms  of the Richardson extrapolation level. To achieve a relative error below $1\%$, ASGQ  combined with numerical smoothing and level-$1$  Richardson extrapolation	significantly outperforms  the other methods. (a) Digital option under GBM, (b)  digital option under Heston.}
	\label{fig: Digital option under GBM and Heston: Computational work comparison for the different methods with the different configurations in terms of the level of Richardson extrapolation. To achieve a relative error below 1,
		the ASGQ method combined with numerical smoothing and level 1 of Richardson extrapolationsignificantly outperforms  the other methods.}
\end{figure}
\begin{figure}[h!]
	\centering 
	\begin{subfigure}{0.4\textwidth}
		\includegraphics[width=\linewidth]{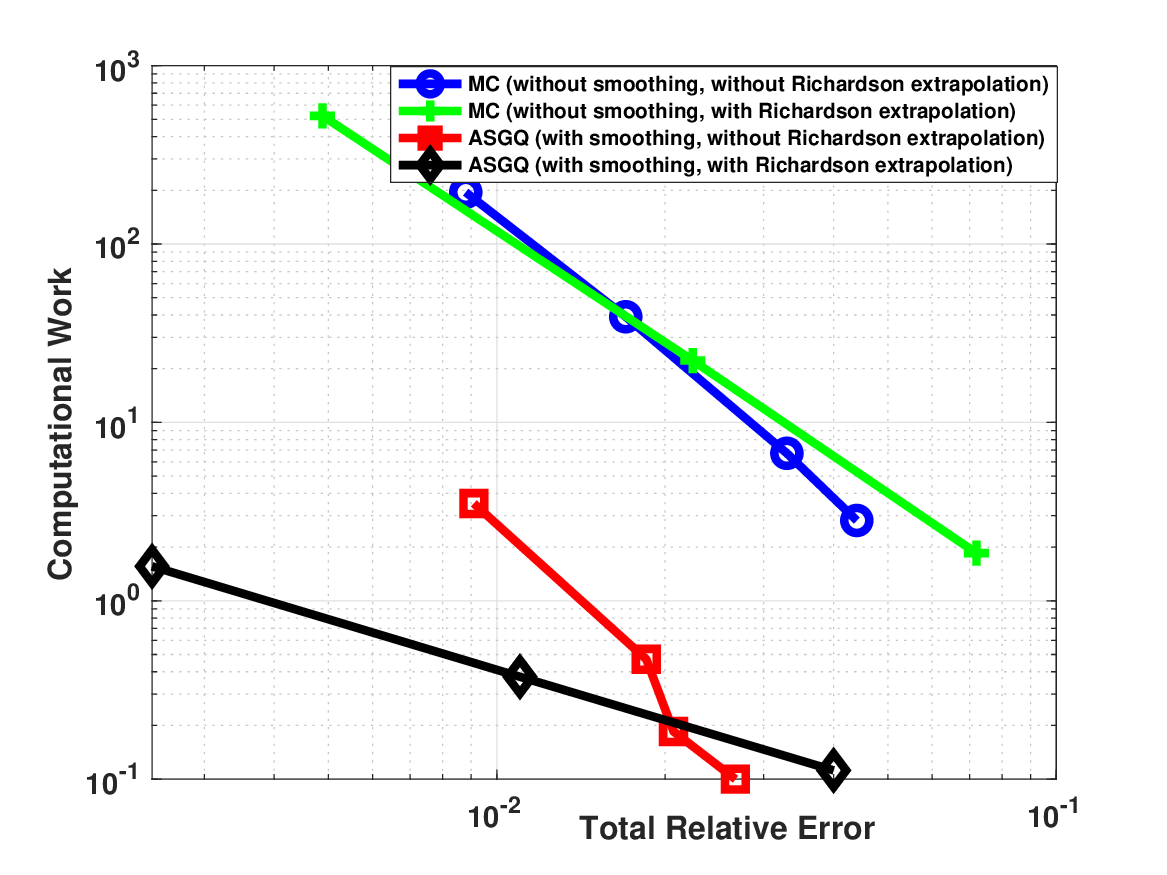}
		\caption{}
		\label{fig:1}
	\end{subfigure}\hfil 
	\begin{subfigure}{0.4\textwidth}
		\includegraphics[width=\linewidth]{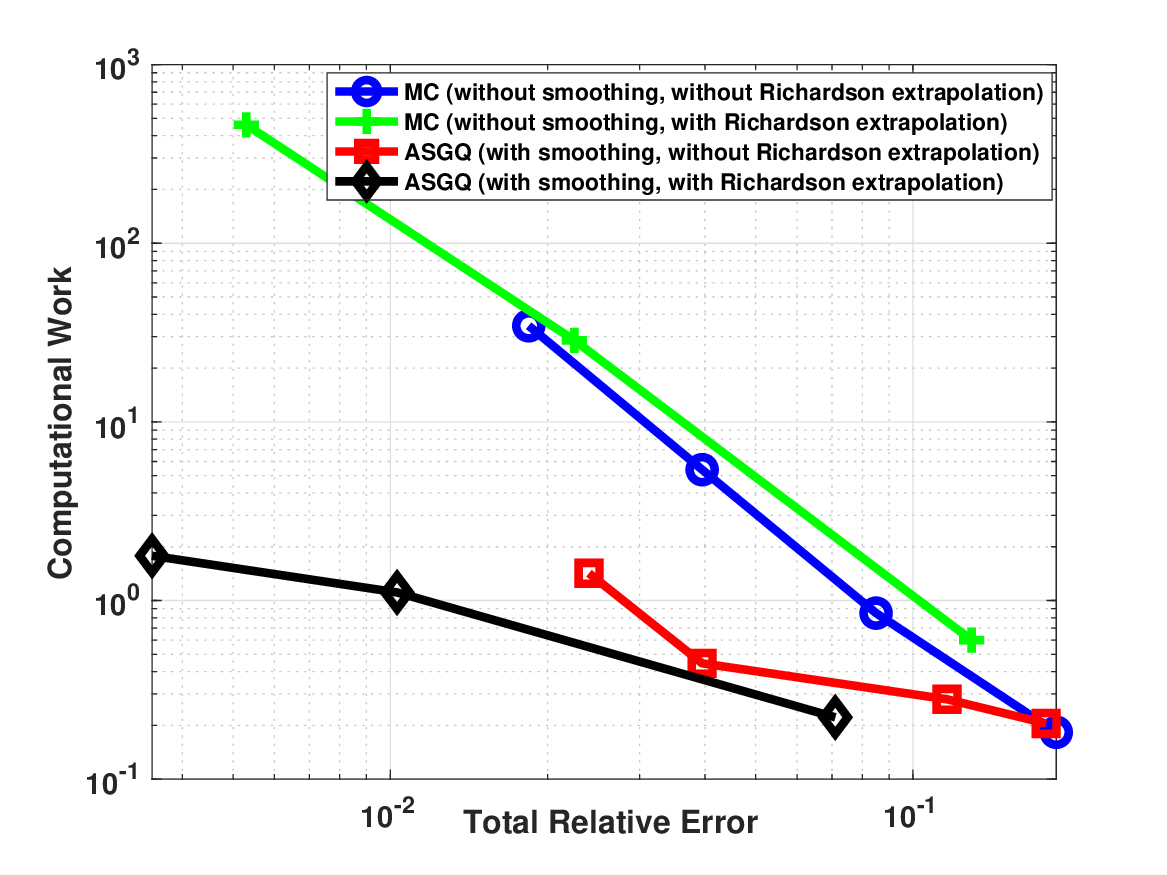}
		\caption{}
		\label{fig:2}
	\end{subfigure}
	\caption{Computational work comparison for the different methods with the various configurations in terms  of the Richardson extrapolation level. To achieve a relative error below $1\%$,	 ASGQ  combined with numerical smoothing and level-$1$  Richardson extrapolation
		significantly outperforms  the other methods. (a) Call option under GBM, (b)  call option under Heston.}
	\label{fig: Call option under GBM and Heston: Computational work comparison for the different methods with the different configurations in terms of the level of Richardson extrapolation. To achieve a relative error below 1,		 ASGQ  combined with numerical smoothing and level 1 of Richardson extrapolation
		significantly outperforms  the other methods.}
\end{figure}
\begin{figure}[h!]
	\centering
	\includegraphics[width=0.4\linewidth]{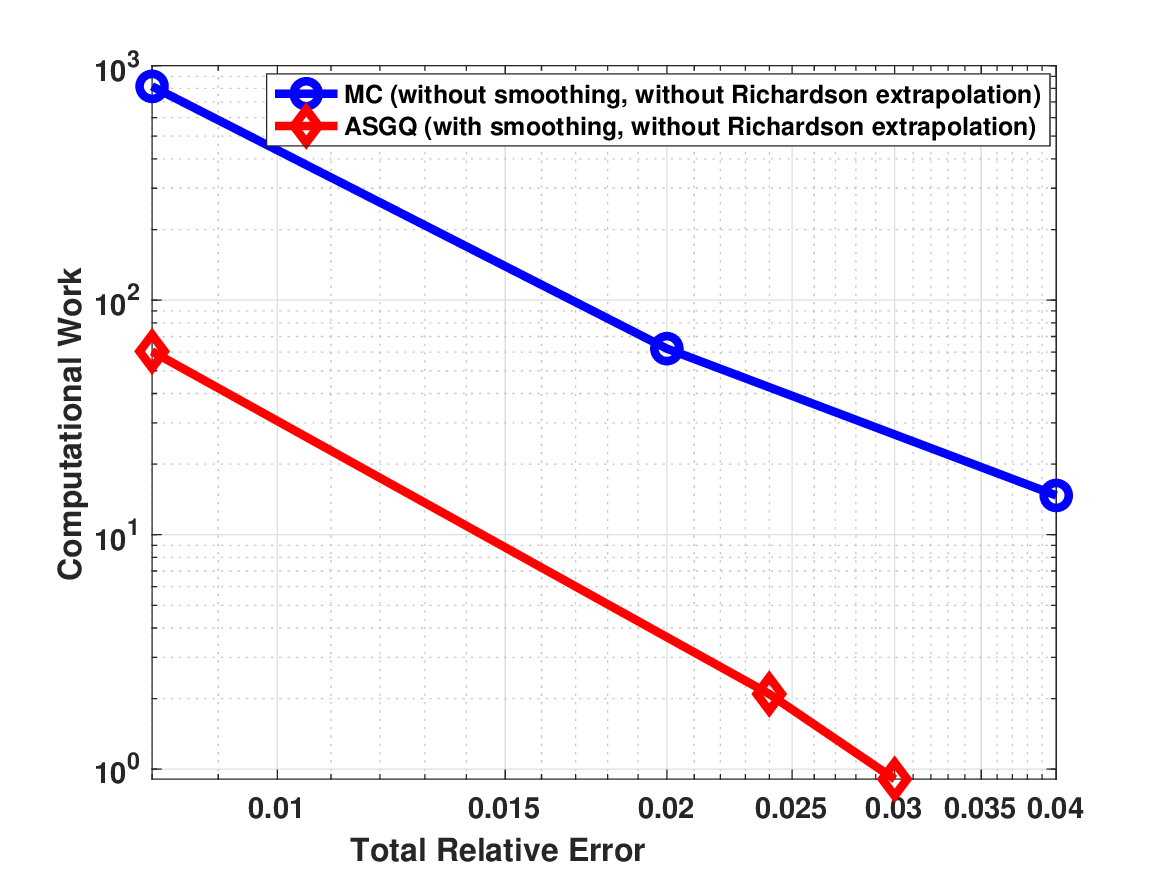}
	\caption{Four-asset basket call option under GBM: Computational work comparison for the different methods. To achieve a relative error below $1\%$, ASGQ  combined with numerical smoothing significantly outperforms the  MC method.}
	\label{fig: 4-dimensional basket call option under GBM: Computational work comparison for the different methods. To achieve a relative error below 1, ASGQ  combined with numerical smoothing significantly outperforms the  MC method.}
\end{figure}
\begin{remark}[Regarding rQMC with numerical smoothing]
We also combined numerical smoothing with the  rQMC method and observed an improvement in the performance compared to the case without smoothing (Section \ref{sec: QMC with numerical smoothing versus QMC without  smoothing}). Moroever, the rQMC method with numerical smoothing consistently outperforms the MC method to achieve a relative error below $1\%$. However,  we consistently observe  that the ASGQ method  outperforms the  rQMC method in all our numerical examples,  when both are combined with numerical smoothing. In particular, as an illustration, Figure \ref{fig: Digital option under GBM: Computational work comparison for the different methods (MC, QMC and ASGQ) with the different configurations in terms of the level of Richardson extrapolation. To achieve a relative error below $1$, the ASGQ method combined with numerical smoothing and level 1 of Richardson extrapolation	significantly outperforms  the other methods.}  shows the comparison for the example of the digital option under the GBM model. 
	\begin{figure}[h!]
		\centering
		\includegraphics[width=0.4\linewidth]{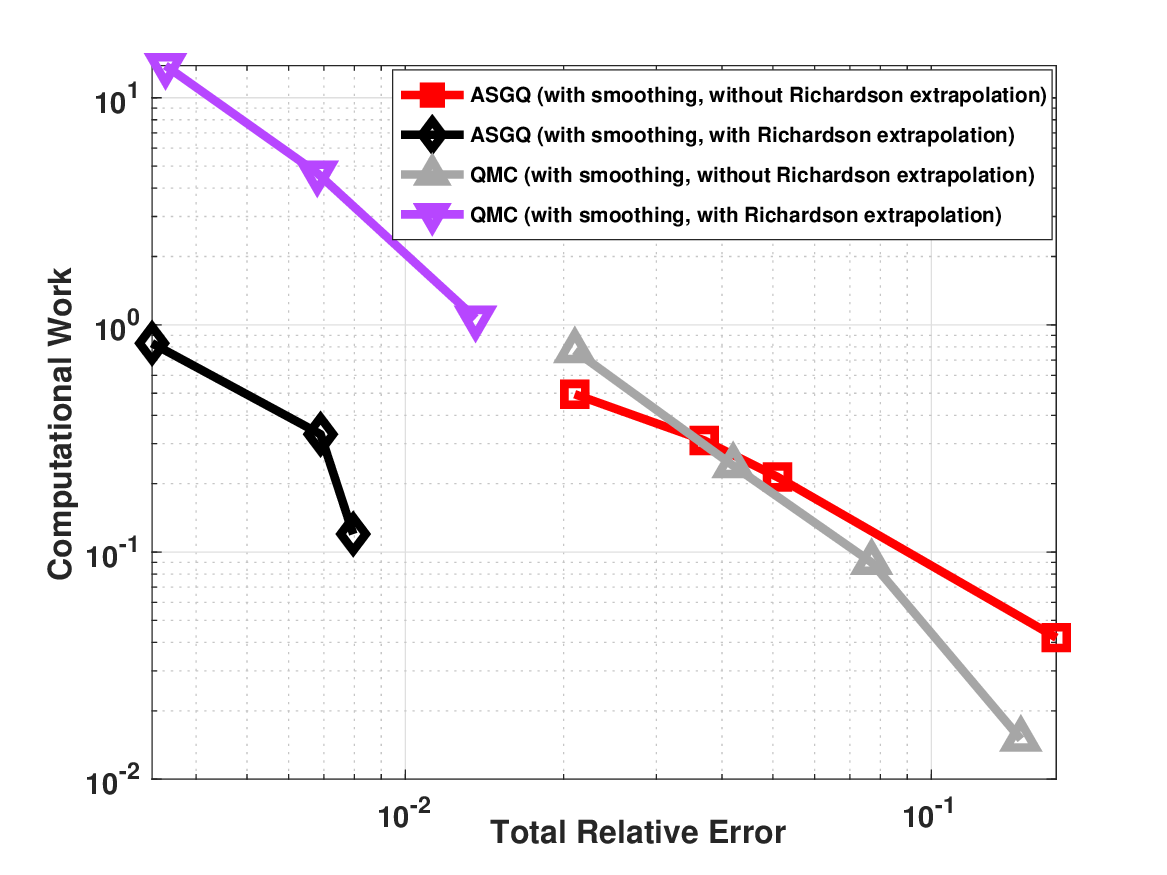}
		\caption{Digital option under GBM: Computational work comparison of  rQMC and ASGQ, both combined with numerical smoothing, with  different  Richardson extrapolation level configurations.  To achieve a relative error below $1\%$,  ASGQ  combined with numerical smoothing and level-$1$  Richardson extrapolation	significantly outperforms  the other methods.}
	\label{fig: Digital option under GBM: Computational work comparison for the different methods (MC, QMC and ASGQ) with the different configurations in terms of the level of Richardson extrapolation. To achieve a relative error below $1$, the ASGQ method combined with numerical smoothing and level 1 of Richardson extrapolation	significantly outperforms  the other methods.}
	\end{figure}
	\end{remark}

\newpage
\

\textbf{Acknowledgments}  C. Bayer gratefully acknowledges support from the German Research Foundation (DFG) via the Cluster of Excellence MATH+ (project AA4-2) and the individual grant BA5484/1. This publication is based on  work supported by the King Abdullah University of Science and Technology (KAUST) Office of Sponsored Research (OSR) under Award No. OSR-2019-CRG8-4033 and the Alexander von Humboldt Foundation.  The authors  are  incredibly grateful to  the anonymous referees for their valuable comments and suggestions  that greatly contributed to shaping the final version of the paper.


\bibliographystyle{plain}
\bibliography{smoothing} 

\appendix

\section{Details of the Proof of Theorem \ref{thr:smoothness} in Section \ref{sec:Analiticity Analysis}}\label{appendix:Details for the proof of Theorem}
In this section, we state and prove the theoretical results  for the proof of Theorem \ref{thr:smoothness} in Section \ref{sec:Analiticity Analysis}. We use the same notation as  in Section \ref{sec:Analiticity Analysis}.  In particular,  we recall that $X^N_T$ denotes the  numerical solution of the SDE \eqref{eq:SDE}   using the   Euler–Maruyama  scheme along  the grid $\mathcal{D}^N$.
\begin{lemma}
  \label{lem:dXdZ}
 If  Assumption \ref{ass:boundedness-derivative} holds, we have the following:
  \begin{equation*}
    \f{\pa X^N_T}{\pa z_{n,k}}(Z_{-1}, \textbf{Z}^N) = 2^{-n/2+1} \mathcal{O}(1)
  \end{equation*}
  in the sense that the $\mathcal{O}(1)$ term does not depend on $n$ or $k$.
\end{lemma}
\begin{proof}
  First, Assumption~\ref{ass:boundedness-derivative} implies
  that $\f{\pa X^N_T}{\pa \Delta^N_\ell W} = \mathcal{O}(1)$. Indeed, 
  \begin{equation*}
    \f{\pa X^N_T}{\pa \Delta^N_\ell W} = \f{\pa X^N_T}{\pa X^N_{\ell+1}}
    \f{\pa X^N_{\ell+1}}{\pa \Delta^N_\ell W} = \mathcal{O}(1) b(X^N_\ell) =
    \mathcal{O}(1).
  \end{equation*}
  Next, we must  identify the increments $\Delta^N_\ell$ that depend on
  $Z_{n,k}$. This is the case if and only if (iff)  the support of $ \psi_{n,k}$ has a nonempty
  intersection with $]t^N_{\ell}, t^N_{\ell+1}[$. Explicitly, this means that
  \begin{equation*}
    \ell 2^{-(N-n+1)} -1 < k < (\ell+1) 2^{-(N-n+1)}.
  \end{equation*}
  If we fix $N$, $k$, and  $n$, the derivative of $\Delta^N_\ell W$
  w.r.t.~$Z_{n,k}$ does not vanish iff $
    2^{N-n+1} k \le \ell < 2^{N-n+1} (k+1)$  because
  \begin{equation}\label{eq:dWdZ}
    \abs{\f{\pa \Delta^N_\ell W}{\pa Z_{n,k}}} = \abs{\Delta^N_\ell \Psi_{n,k}}
    \le 2^{-(N-n/2)}.
  \end{equation}
Thus, we obtain the following:
  \begin{equation}
    \label{eq:2}
    \f{\pa X^N_T}{\pa z_{n,k}}(Z_{-1}, \textbf{Z}^N) =
    \sum_{\ell=2^{N-n+1}k}^{2^{N-n+1}(k+1)-1} \f{\pa X^N_T}{\pa \Delta^N_\ell
      W} \f{\pa \Delta^N_\ell W}{\pa Z_{n,k}} = 2^{N-n+1} 2^{-(N-n/2)}
    \mathcal{O}(1) = 2^{-n/2+1} \mathcal{O}(1).\qedhere
  \end{equation}
\end{proof}
\begin{lemma}
  \label{lem:d2XdZdY}
If  Assumption \ref{ass:boundedness-derivative} holds, then  similar to Lemma~\ref{lem:dXdZ}, we have
  \begin{equation*}
    \f{\pa^2 X^N_T}{\pa y \pa z_{n,k}}(Z_{-1}, \textbf{Z}^N) = 2^{-n/2+1} \mathcal{O}(1).
  \end{equation*}
\end{lemma}
\begin{proof}
  $\Delta^N_\ell W$ is a linear function in $Z_{-1}$ and $\textbf{Z}^N$, implying that
  all mixed derivatives $\f{\pa^2\Delta^N_\ell W}{\pa Z_{n,k} \pa Z_{-1}}$ vanish.
  From equation~\eqref{eq:2} we hence obtain
  \begin{equation*}
    \f{\pa^2 X^N_T}{\pa z_{n,k} \pa y}(Z_{-1}, \textbf{Z}^N) =
    \sum_{\ell=2^{N-n+1}k}^{2^{N-n+1}(k+1)-1} \f{\pa^2 X^N_T}{\pa \Delta^N_\ell
      W \pa Z_{-1}} \f{\pa \Delta^N_\ell W}{\pa Z_{n,k} }.
  \end{equation*}
  Further,
  \begin{equation*}
    \f{\pa^2 X^N_T}{\pa \Delta^N_\ell W \pa Z_{-1}} = \sum_{j=0}^{2^{N+1}-1}
    \f{\pa^2 X^N_T}{\pa \Delta^N_\ell W \pa \Delta^N_j W} \f{\pa \Delta^N_j W}{\pa
      Z_{-1}}.
  \end{equation*}
  Note that
  \begin{equation}
    \label{eq:d2XdWdW}
    \f{\pa^2 X^N_T}{\pa \Delta^N_\ell W \pa \Delta^N_j W} = \f{\pa^2
      X^N_T}{\pa X^N_{\ell+1} \pa X^N_{j+1}} b(X^N_\ell) b(X^N_j) + \indic{j
      < \ell} \f{\pa X^N_T}{\pa X_\ell^N} b^\prime(X^N_\ell) \f{\pa
      X^N_\ell}{\pa X^N_{j+1}} b(X^N_j) = \mathcal{O}(1)
  \end{equation}
  using Assumption~\ref{ass:boundedness-derivative}. We also have $\f{\pa \Delta^N_j W}{\pa
    Z_{-1}} = \mathcal{O}(2^{-N})$, implying the statement of the lemma.
\end{proof}
\begin{remark}
  Lemmas~\ref{lem:dXdZ} and~\ref{lem:d2XdZdY} also hold (mutatis mutandis) for
  $z_{n,k} = y$ (with $n = 0$).  
\end{remark}
\begin{proposition}
  \label{prop:first-derivatives}
  if Assumptions \ref{ass:boundedness-derivative} and \ref{ass:boundedness-inverse} hold, we have $\f{\pa H^N(\textbf{z}^N)}{\pa z_{n,k}} = \mathcal{O}(2^{-n/2})$ such
  that the constant in front of $2^{-n/2}$ does not depend on $n$ or $k$.
\end{proposition}
\begin{proof}
 if Assumptions \ref{ass:boundedness-derivative} and \ref{ass:boundedness-inverse} hold, we obtain the following:
  \begin{align*}
    \f{\pa H^N(\textbf{z}^N)}{\pa z_{n,k}} &= -\int_\R g\left( X^N_T
      (y, \textbf{z}^N) \right) \f{\pa}{\pa y} \left[ \left( \f{\pa X^N_T}{\pa y}(y, \textbf{z}^N)
  \right)^{-1} \f{\pa X^N_T}{\pa z_{n,k}}(y, \textbf{z}^N)  \f{1}{\sqrt{2\pi}}
                                    e^{-\f{y^2}{2}} \right] dy\\
    &= -\int_\R g\left( X^N_T(y, \textbf{z}^N) \right) \Biggl[- \left( \f{\pa
      X^N_T}{\pa y}(y, \textbf{z}^N) \right)^{-2} \f{\pa^2 X^N_T}{\pa y^2}(y, \textbf{z}^N)
      \f{\pa X^N_T}{\pa z_{n,k}}(y, \textbf{z}^N) +\\
    &\quad+ \left( \f{\pa X^N_T}{\pa y}(y, \textbf{z}^N) \right)^{-1} \f{\pa^2
      X^N_T}{\pa z_{n,k} \pa y}(y, \textbf{z}^N) - y \left( \f{\pa X^N_T}{\pa y}(y, \textbf{z}^N)
  \right)^{-1} \f{\pa X^N_T}{\pa z_{n,k}}(y, \textbf{z}^N) \Biggr] \f{1}{\sqrt{2\pi}}
                                    e^{-\f{y^2}{2}}dy.
  \end{align*}
  Hence, Lemmas~\ref{lem:dXdZ} and ~\ref{lem:d2XdZdY}  with
  Assumption~\ref{ass:boundedness-inverse} (for $p=2$) imply that
  \begin{equation*}
    \f{\pa H^N(\textbf{z}^N)}{\pa z_{n,k}} = \mathcal{O}(2^{-n/2}),
  \end{equation*}
  with constants independent of $n$ and $k$.\footnote{When $F^N(Z_{-1}, \textbf{Z}^N) = \mathcal{O}(c)$ for some deterministic
  			constant $c$,  this property is retained when integrating out one of the rdvs, i.e., we still achieve   $		\int_\R F^N(y, \textbf{Z}^N) \f{1}{\sqrt{2\pi}} e^{-\f{y^2}{2}}dy = \mathcal{O}(c).$}
\end{proof}
For the general case we require the following lemma.
\begin{lemma}
  \label{lem:d2XdZ2}
  For any $p \in \N$ and indices $n_1, \ldots, n_p$ and $k_1, \ldots, k_p$
  (satisfying $0 \le k_j < 2^{n_j}$) we have (with constants independent of
  $n_j, k_j$)
  \begin{equation*}
    \f{\pa^p X^N_T}{\pa z_{n_1,k_1} \cdots \pa z_{n_p,k_p}} (Z_{-1}, \textbf{Z}^N) =
    \mathcal{O}\left( 2^{-\sum_{j=1}^p n_j/2} \right).
  \end{equation*}
  The result also holds (mutatis mutandis) if one or several $z_{n_j,k_j}$ are
  replaced by $y = z_{-1}$ (with $n_j$ set to $0$).
\end{lemma}
\begin{proof}
 Each $\Delta^N_\ell W$ is a linear function of
  $(Z_{-1}, \textbf{Z}^N)$ implying that all higher derivatives of $\Delta^N_\ell W$
  w.r.t.~$(Z_{-1}, \textbf{Z}^N)$ vanish. Hence,
  \begin{equation*}
    \f{\pa^p X^N_T}{\pa Z_{n_1,k_1} \cdots \pa Z_{n_p,k_p}} = \sum_{\ell_1 =
      2^{N-n_1+1} k_1}^{2^{N-n_1+1}(k_1+1) - 1} \cdots \sum_{\ell_p =
      2^{N-n_p+1} k_p}^{2^{N-n_p+1}(k_p+1) - 1} \f{\pa^p X^N_T}{\pa
      \Delta^N_{\ell_1} \cdots \pa \Delta^N_{\ell_p} W} \f{\pa
      \Delta^N_{\ell_1} W}{\pa Z_{n_1,k_1}} \cdots \f{\pa
      \Delta^N_{\ell_p}W}{\pa Z_{n_p,k_p}}.
  \end{equation*}
 By an argument similar to that made for \eqref{eq:d2XdWdW}, we obtain
  \begin{equation*}
    \f{\pa^p X^N_T}{\pa \Delta^N_{\ell_1} \cdots \pa \Delta^N_{\ell_p} W} = \mathcal{O}(1).
  \end{equation*}
  By~\eqref{eq:dWdZ}, we observe that each summand in the aforementioned sum is of order
  $\prod_{j=1}^p 2^{-(N-n_j/2)}$. The number of summands in total is
  $\prod_{j=1}^p 2^{N-n_j+1}$. Therefore, we obtain  the desired result. 
\end{proof}
\begin{proof}[Sketch of the proof of Theorem~\ref{thr:smoothness}]
	We apply integration by parts $p$ times, as performed in the proof of
	Proposition~\ref{prop:first-derivatives}, which shows that we can again
	replace the mollified payoff function $g_\delta$ by the true, nonsmooth function
	$g$. Moreover,  using this procedure, we obtain a formula of the form
		\begin{equation*}
			\f{\pa^p H^N}{\pa z_{n_1,k_1} \cdots \pa z_{n_p,k_p}} (z^N) = \int_{\R}
			g\left( X^N_T(y, z^N) \right) \blacksquare \f{1}{\sqrt{2\pi}} e^{-\f{y^2}{2}} dy,
	\end{equation*}
	where $\blacksquare$ represents a long sum of products of various
	terms. However,   when  the
	derivatives w.r.t.~$y$ are ignored, each summand contains all derivatives
	w.r.t.~$z_{n_1,k_1}, \ldots, z_{n_p,k_p}$ exactly once. (Generally,
	each summand is a product of the derivatives of $X^N_T$ w.r.t.~some
	$z_{n_j,k_j}$s, possibly including other terms, such as polynomials in $y$ and
	derivatives w.r.t.~$y$.) As all other terms are assumed to be of
	order $\mathcal{O}(1)$ based on Assumptions~\ref{ass:boundedness-derivative}
	and~\ref{ass:boundedness-inverse}, the result suggested by 	Lemma~\ref{lem:d2XdZ2} is implied, concluding the proof of Theorem \ref{thr:smoothness}.
\end{proof}

\section{Discussion of Assumption \ref{ass:boundedness-inverse}}\label{appendix:Discussion of Assumption}
We present sufficient conditions for   Assumption \ref{ass:boundedness-inverse} to be valid in the one-dimensional setting. Moreover, we  discuss its limitation and some multivariate cases in which this assumption holds.

We want to examine the term given by $\left(\frac{\partial X_T^N}{\partial y} \left( Z_{-1}, \mathbf{Z}^N\right) \right)^{-p}$  for $p \in \nset$.  For this, we  consider the one-dimensional SDE
\begin{equation*}\label{eq:one_dim_SDE_example}
dX_t=a(X_t) dt +b(X_t) dW_t.
\end{equation*}  
For  ease of presentation, we set the drift term $a(.)$ to  zero. Moreover, using the Brownian bridge construction, we achieve 
\begin{equation}\label{eq:one_dim_SDE_example_bridge}
dX_t=b(X_t)  \left(\frac{y}{\sqrt{T} }dt+ dB_t \right),
\end{equation}
where  $y$  is  a  standard Gaussian rdv and  $B$ is the Brownian bridge.

The solution of \eqref{eq:one_dim_SDE_example_bridge}, at the final time  $T>0$  is 
\begin{equation*}
X_T=x_0+   \frac{y}{\sqrt{T} } \int_{0}^{T} b(X_s) ds+  \int_{0}^{T} b(X_s)   dB_s,
\end{equation*}
and consequently,
\begin{equation*}
\frac{\partial X_T}{\partial y}=    \frac{y}{\sqrt{T} } \int_{0}^{T} b^\prime(X_s)   \frac{\partial X_s}{\partial y } ds+   \frac{1}{\sqrt{T} } \int_{0}^{T} b(X_s) ds+  \int_{0}^{T} b^\prime(X_s)   \frac{\partial X_s}{\partial y }   dB_s.
\end{equation*}  
This implies that $\frac{\partial X_T}{\partial y}$ solves 
\begin{equation*}   
	\begin{cases}
	d \left(\frac{\partial X_T}{\partial y} \right) &=    \frac{b(X_t)}{\sqrt{T}}  dt+  b^\prime(X_t)   \frac{\partial X_t}{\partial y }   dW_t,\\
	\frac{\partial X_T}{\partial y}\mid_{t=0}&=0.
	\end{cases}
\end{equation*}
Using Duhamel's principle, we obtain
\begin{equation*}
\frac{\partial X_T}{\partial y}=  \int_{0}^{T}  \frac{b(X_s)}{\sqrt{T}} \; \exp \left(  \left( \int_{s}^{T} b^\prime(X_u)dW_u \right)  -\frac{1}{2}  \int_{s}^{T}  (b^\prime)^2 (X_u) du   \right) ds.
\end{equation*}
If there exists  $b_0 \in \rset $ such that 
	\begin{equation}\label{eq:cond1 asummptionA_3}
	b^2(x) \ge b_0^2, \quad \forall x \in \rset,
	\end{equation}
 then 
\begin{align*}
\abs{\frac{\partial X_T}{\partial y}}  &\ge \frac{\abs{b_0}}{\sqrt{T}}          \int_{0}^{T}  \exp \left(  \left( \int_{s}^{T} b^\prime(X_u)dW_u \right)  -\frac{1}{2}  \int_{s}^{T}  (b^\prime)^2 (X_u) du   \right) ds,\\
&\ge \frac{\abs{b_0}}{\sqrt{T}}        \exp \left(  \left( \int_{0}^{T} b^\prime(X_u)dW_u \right)  -\frac{1}{2}  \int_{0}^{T}  (b^\prime)^2 (X_u) du   \right) ds,
\end{align*}
and consequently,  for any $p \in \nset$, we  obtain
\begin{align*}
\left(\abs{\frac{\partial X_T}{\partial y} } \right)^{-p}  	\le  \left(\frac{\abs{b_0}}{\sqrt{T}}\right)^{-p}        \exp \left( -p  \left( \left( \int_{0}^{T} b^\prime(X_u)dW_u \right)  -\frac{1}{2}  \int_{0}^{T}  (b^\prime)^2 (X_u) du\right)   \right) ds,
\end{align*}
and the sufficient condition for Assumption \ref{ass:boundedness-inverse} to be   valid is that for any $p \in \nset$, there exists  a real deterministic constant  $D_p>0$ such that
\begin{align}\label{eq:condition inverse moments}
		\expt{  \exp \left( -p  \left(\left( \int_{0}^{T} b^\prime(X_u)dW_u \right)  -\frac{1}{2}  \int_{0}^{T}  (b^\prime)^2 (X_u) du   \right) \right)  }  &\le D_p.
\end{align}
For the particular one-dimensional  GBM model, condition \eqref{eq:condition inverse moments} is clearly satisfied. Moreover,  both  (i)  one-dimensional models with a linear or constant diffusion and  (ii)  multivariate models with a linear drift and constant diffusion satisfy Assumption \ref{ass:boundedness-inverse}. Interestingly, the multivariate lognormal model  can  be observed in case (ii) (refer to \cite{bayersmoothing} for further details).  However,  there may be  cases  in which Assumption \ref{ass:boundedness-inverse} is not fulfilled, e.g.,  $X_T=W_T^2$, corresponding to a system of SDEs where   the diffusion coefficient does not satisfy condition \eqref{eq:cond1 asummptionA_3}. Nevertheless, the proposed method  works well in such cases because (using notation of Section \ref{sec:Discrete time, practical motivation}) $g(X_T)=G(y_1^2)$, and then we can apply our numerical smoothing technique  to obtain a highly smooth integrand. Finally, an additional investigation on the sufficient conditions for our smoothness Theorem  \ref{thr:smoothness} to be valid in high dimensions is an open problem and is not within the scope of this work. 

\section{More Details on the Work Discussion of the ASGQ Method}\label{appendix:More details on the error and work discussion of ASGQ method combined with numerical smoothing}
Under certain conditions of the regularity parameters $p$  and $s$, we can achieve  $\text{Work}_{\text{ASGQ}}=\Ordo{\text{TOL}^{-1}}$ under the best scenario ($p,s \gg  1$). In fact, let q=p/2, then using  the method of Lagrange multipliers, we obtain
\begin{equation*}
M_{ASGQ}\propto\Delta t^{\frac{2q+s-qs}{q(qs+2q+s)}},\quad \text{and} \quad  M_{\text{lag}}\propto\Delta t^{\frac{2q+s-qs}{\frac{s}{2}(qs+2q+s)}}.
\end{equation*}
Using the constraint in \eqref{eq:opt_ASGQ_work}, we can easily demonstrate that, for an error tolerance $\text{TOL}$, we achieve $\Delta t=\Ordo{\text{TOL}^{\frac{qs+2q+s}{qs-2q-s}}}$. Therefore, the optimal work, $\text{Work}_{\text{ASGQ}}$, solution of \eqref{eq:opt_ASGQ_work} satisfies
\begin{align*}
\text{Work}_{\text{ASGQ}} \propto    \Delta t^{-1} \times M_{\text{ASGQ}} \times M_{\text{Lag}}&\propto   \Delta t^{-1}  \Delta t^{\frac{2q+s-qs}{q(qs+2q+s)}} \Delta t^{\frac{2q+s-qs}{\frac{s}{2}(qs+2q+s)}} \\
&\propto \text{TOL}^{-1-\frac{2(2q+s)}{qs-2q-s}-\frac{1}{q}-\frac{2}{s}}\\
&=\Ordo{\text{TOL}^{-1}},\quad \text{because} \: p,s \gg  1. 
\end{align*}
\section{Simulation Schemes for   the Heston Dynamics}\label{sec:Schemes to simulate the Heston dynamics}
\subsection{Modified Euler scheme}\label{sec:Discretization of Heston model with a non smooth transformations for the volatility process}
The forward Euler scheme can be used to simulate the Heston model. The literature has reported many solutions to avoid the problems arising from the use of negative values of the volatility process $v_t$ in \eqref{eq:dynamics Heston}  \cite{lord2010comparison}. Table \ref{Numerical schemes for CIR process} introduces $f_1, f_2$, and $f_3$, which imply various schemes when different choices are adopted. The forward Euler scheme to discretize \eqref{eq:dynamics Heston}  yields
\begin{align*}
\hat{S}_{t+\Delta t}&= \hat{S}_t +\mu \hat{S}_t \Delta t+\sqrt{\hat{V}_t \Delta t} \hat{S}_t Z_s \nonumber\\
\hat{V}_{t+\Delta t}&=f_1(\hat{V}_t)+\kappa (\theta-f_2(\hat{V}_t)) \Delta t+\xi \sqrt{f_3(\hat{V}_t) \Delta t} Z_V \nonumber\\
\hat{V}_{t+\Delta t}&=f_3(\hat{V}_{t+\Delta t}) \COMMA
\end{align*}
where $Z_s$ and $Z_V$  are two correlated standard normal rdvs with correlation $\rho$. 
\begin{table}[h!]
	\centering
	\begin{tabular}{l*{6}{c}r}
		\toprule[1.5pt]
	Scheme &  $f_1$& $f_2$  & $f_3$     \\
	\hline
	Full truncation scheme & $\hat{V}_t$ &  $\hat{V}_t^+$&$\hat{V}_t^+$\\
	Partial truncation scheme & $\hat{V}_t$ &  $\hat{V}_t$&$\hat{V}_t^+$\\
	Reflection scheme  &$\abs{\hat{V}_t}$ & $\abs{\hat{V}_t}$& $\abs{\hat{V}_t}$\\
			\bottomrule[1.25pt]
	\end{tabular}
	\caption{Different variants for the  forward Euler scheme for the Heston model. $\hat{V}_t^+=\max(0,\hat{V}_t)$.}
	\label{Numerical schemes for CIR process}
\end{table}

Lord et al. \cite{lord2010comparison} suggested that the full truncation scheme is an optimal option in terms of  the weak error convergence. Therefore,  we used this variant of the forward Euler scheme. 
\subsection{Moment-matching scheme}\label{sec:Euler schemes with moment matching}
We consider the moment-matching  scheme  suggested by Andersen and Brotherton-Ratcliffe \cite{andersen2005extended} (the ABR scheme). This scheme assumes that the variance $v_t$ is  locally lognormal, and the parameters are determined such that the first two moments of the discretization coincide with the theoretical moments: 
\begin{align*}\label{eq: vol_moment_matching}
\hat{V}(t+\Delta t)&= \left(e^{-\kappa \Delta t}  \hat{V}(t) + \left( 1-e^{-\kappa \Delta t}\right) \theta  \right) e^{-\frac{1}{2} \Gamma(t)^2 \Delta t+\Gamma(t) \Delta W_v(t)}\nonumber\\
\Gamma^2(t) &=\Delta t^{-1} \log \left(  1+ \frac{\frac{1}{2}\xi^2 \kappa^{-1} \hat{V}(t) (1-e^{-2 \kappa \Delta t})}{\left( e^{-\kappa \Delta t} \hat{V}(t) +(1-e^{- \kappa \Delta t}) \theta \right)^2 }\right).
\end{align*}
As reported in \cite{lord2010comparison}, the scheme is  easy to implement and   more effective than many of the Euler variants presented in Section \ref{sec:Discretization of Heston model with a non smooth transformations for the volatility process}; however, this scheme exhibits  a nonrobust weak error behavior w.r.t.~the model parameters.

\subsection{Heston OU-based scheme}\label{sec:Discretization of Heston model with the volatility process Simulated using the sum of  Ornstein-Uhlenbeck (Bessel) processes}

Because any  OU process is normally distributed, the sum of $n$ squared OU processes is  chi-squared distributed with $n$ degrees of freedom, where $n \in \nset_+$.  We define $\mathbf{X}$ as a  $n$-dimensional vector-valued OU process with
\begin{equation}\label{equivalent OU process} 
\mathrm{d}X_t^i = \alpha X_t^i \mathrm{d}t + \beta \mathrm{d}W_t^i,
\end{equation}
where $\mathbf{W}$ is an $n$-dimensional vector of independent Brownian motions. 

We also define  the process $Y_t$ as   follows:
\begin{equation*}
Y_t = \sum_{i = 1}^n \left( X_t^i \right)^2.
\end{equation*}
Then, because
\begin{equation*}
\mathrm{d} \left( X_t^i \right)^2 =  2 X_t^i \mathrm{d}X_t^i + 2 \mathrm{d} \langle X^i \rangle_t  =  \left( 2 \alpha \left( X_t^i \right)^2 + \beta^2 \right) \mathrm{d}t + 2 \beta X_t^i \mathrm{d}W_t^i,
\end{equation*}
we can write (using the independence of the Brownian motions):
\begin{equation}\label{eq:expressing CIR processes from OU processes}
\mathrm{d}Y_t =  \mathrm{d} \left( \sum_{i = 1}^n \left( X_t^i \right)^2 \right)  =  \sum_{i = 1}^n \mathrm{d} \left( X_t^i \right)^2  =  \left( 2 \alpha Y_t + n \beta^2 \right) \mathrm{d}t + 2 \beta \sum_{i = 1}^n X_t^i \mathrm{d}W_t^i.
\end{equation}
Furthermore, the process $ Z_t = \int_0^t \sum_{i = 1}^n X_u^i \mathrm{d}W_u^i$ is a martingale with quadratic variations
\begin{equation*}
\langle Z \rangle_t  =  \int_0^t \sum_{i = 1}^n \left( X_u^i \right)^2 \mathrm{d}u  =  \int_0^t Y_u \mathrm{d}u.
\end{equation*}
Consequently, using the L\'evy characterization theorem, the process $ \widetilde{W}_t = \int_0^t \frac{1}{\sqrt{Y_u}} \sum_{i = 1}^n X_u^i \mathrm{d}W_u^i$ is a Brownian motion. Finally, we obtain 
\begin{eqnarray}\label{eq:equivalent CIR process}
\mathrm{d}Y_t & = & \left( 2 \alpha Y_t + n \beta^2 \right) \mathrm{d}t + 2 \beta \sqrt{Y_t} \mathrm{d}\widetilde{W}_t \nonumber\\
& = & \kappa \left( \theta - Y_t \right) \mathrm{d}t + \xi \sqrt{Y_t} \mathrm{d}W_t,
\end{eqnarray}
where $\kappa = -2 \alpha$, $\theta = -n \beta^2 / 2 \alpha$ and $\xi = 2 \beta$.

Equations \eqref{equivalent OU process}, \eqref{eq:expressing CIR processes from OU processes}, and \eqref{eq:equivalent CIR process} indicate that,   to simulate the   process $Y_t$ given by \eqref{eq:equivalent CIR process}, we can  simulate the  OU process  $\mathbf{X}$ with dynamics  \eqref{equivalent OU process} such that its parameters $(\alpha, \beta)$ are expressed in terms of  those of the  process $Y_t$:
$$ \alpha=-\frac{\kappa}{2},\quad \beta=\frac{\xi}{2}, \quad n=\frac{-2 \theta \alpha}{\beta^2}=\frac{4 \theta  \kappa}{\xi^2}.$$ Consequently, we can simulate the volatility of the Heston model using a sum of OU processes.

\begin{remark}
The previous derivation can  be generalized to cases where $n^\ast$ is not an integer, by considering the time-change of the squared Bessel process (see  Chapter 6 in \cite{jeanblanc2009mathematical} for details). An alternative  method to  generalize the scheme for  any noninteger $n^\ast$ is to consider  $n^\ast=n+p, \: p \in (0,1)$, and   compute   $\expt{g(X_{n^\ast})}$ for any observable $g$ as follows:
\begin{equation*}
\expt{g(X_{n^\ast})}\approx (1-p) \expt{g(X_n)}+ p \expt{g(X_{n+1})}.
\end{equation*}
\end{remark}

\subsection{On the choice of the simulation scheme of the Heston model}\label{sec:On the choice of the simulation scheme of the Heston model}
We determine the optimal  scheme for simulating the Heston model defined in \eqref{eq:dynamics Heston}.  In our setting, an optimal scheme is characterized by two properties: (i) the behavior of mixed rate convergence (Section \ref{sec:Comparison in terms of mixed differences rates}),   which is a critical requirement for the  optimal performance of ASGQ and  (ii) the  weak error behavior (Section \ref{sec:Comparison in terms of  the weak error behavior})  to apply the Richardson extrapolation when necessary. 

Although we tested many parameter sets and obtained consistent numerical observations; for illustration, we only present the results for the single call option based on the Heston model with parameters listed in Table \ref{table: Examples details.}. This   set   corresponds to $n=1$, where $n$ represents the number of OU processes used in the Heston OU-based scheme (Section \ref{sec:Discretization of Heston model with the volatility process Simulated using the sum of  Ornstein-Uhlenbeck (Bessel) processes}).  
%
\subsubsection{Comparison of different schemes in terms of mixed difference rates}\label{sec:Comparison in terms of mixed differences rates}
As emphasized in \cite{haji2016multi,bayer2018hierarchical}, one crucial requirement to achieve the optimal performance of the ASGQ is to check  the  error convergence of the first and mixed difference operators, as expressed by the error contribution $\Delta \text{E}_{\boldsymbol{\beta}}$ in  \eqref{error_contr}. This  is a measure of how much the  quadrature error would decrease after the addition of a new mutli-index $\boldsymbol{\beta}$   to the constructed index set of the  ASGQ estimator, $\mathcal{I}_{\text{ASGQ}}$.  The ASGQ method exhibits optimal behavior if  (i) $\Delta \text{E}_{\boldsymbol{\beta}}$ decreases exponentially fast w.r.t.~$\beta_i$  and (ii) $\Delta \text{E}_{\boldsymbol{\beta}}$ has a  product structure so that  a faster error decay is observed for second differences  compared to the corresponding first difference operators.

In this section, we compare  the three  approaches of simulating Heston dynamics: (i) the full truncation scheme (Section \ref{sec:Discretization of Heston model with a non smooth transformations for the volatility process}),  (ii) the ABR scheme (Section \ref{sec:Euler schemes with moment matching}),  and (iii) the Heston OU-based scheme (Section \ref{sec:Discretization of Heston model with the volatility process Simulated using the sum of  Ornstein-Uhlenbeck (Bessel) processes}) in terms of the mixed difference convergences. In our numerical experiments, we only observe the differences in the mixed difference rates related to the volatility coordinates because we apply schemes that only differ in the way they simulate the volatility process. Figure \ref{fig:first_diff_Heston_call_N_4_set2} illustrates  a comparison of the first difference rates related to the volatility coordinates for  the various schemes. The figure reveals  that the full truncation scheme is the worst scheme and that the Heston OU-based  and  the  ABR schemes  perform very well in terms of the speed of the mixed rate convergence.
\begin{figure}[h!]
	\centering 
	\begin{subfigure}{0.33\textwidth}
		\includegraphics[width=\linewidth]{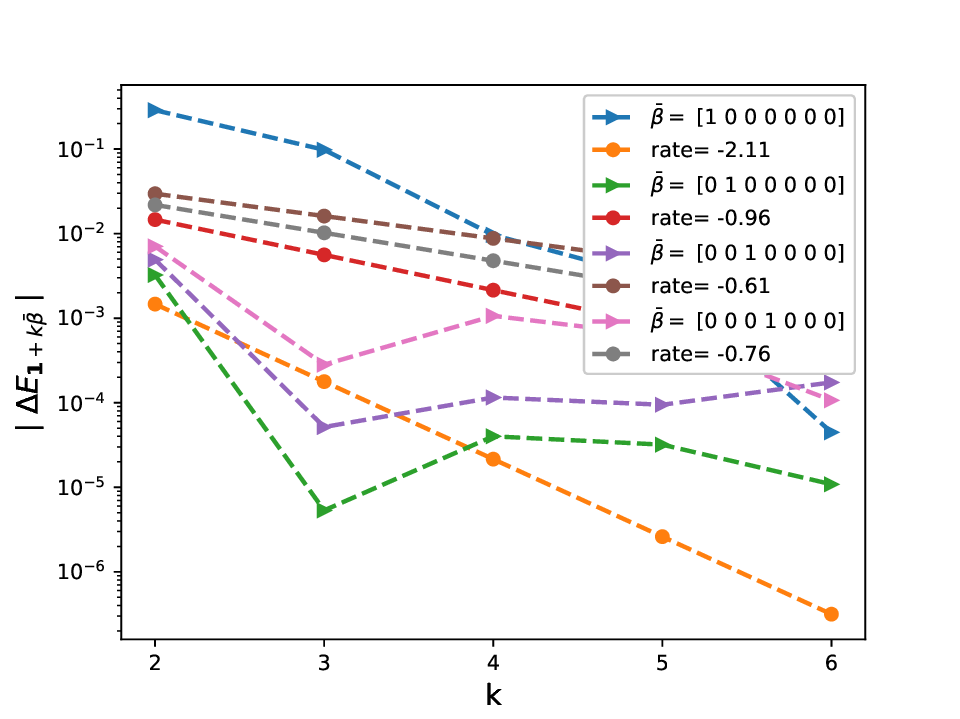}
		\caption{}
		\label{fig:1}
	\end{subfigure}\hfil 
	\begin{subfigure}{0.33\textwidth}
		\includegraphics[width=\linewidth]{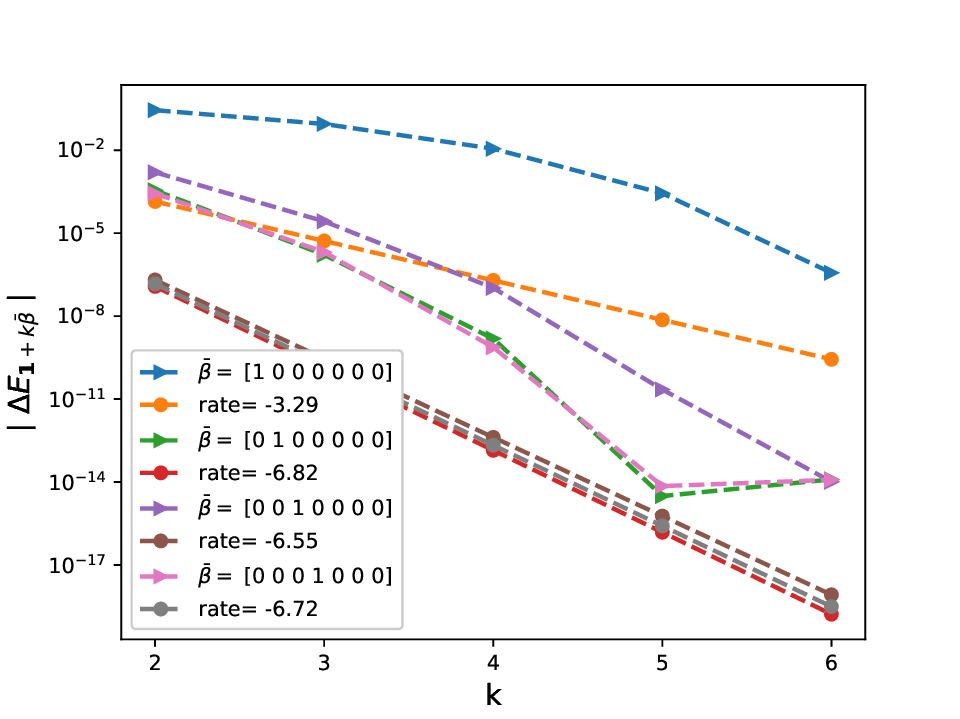}
		\caption{}
		\label{fig:2}
	\end{subfigure}\hfil 
		\begin{subfigure}{0.33\textwidth}
		\includegraphics[width=\linewidth]{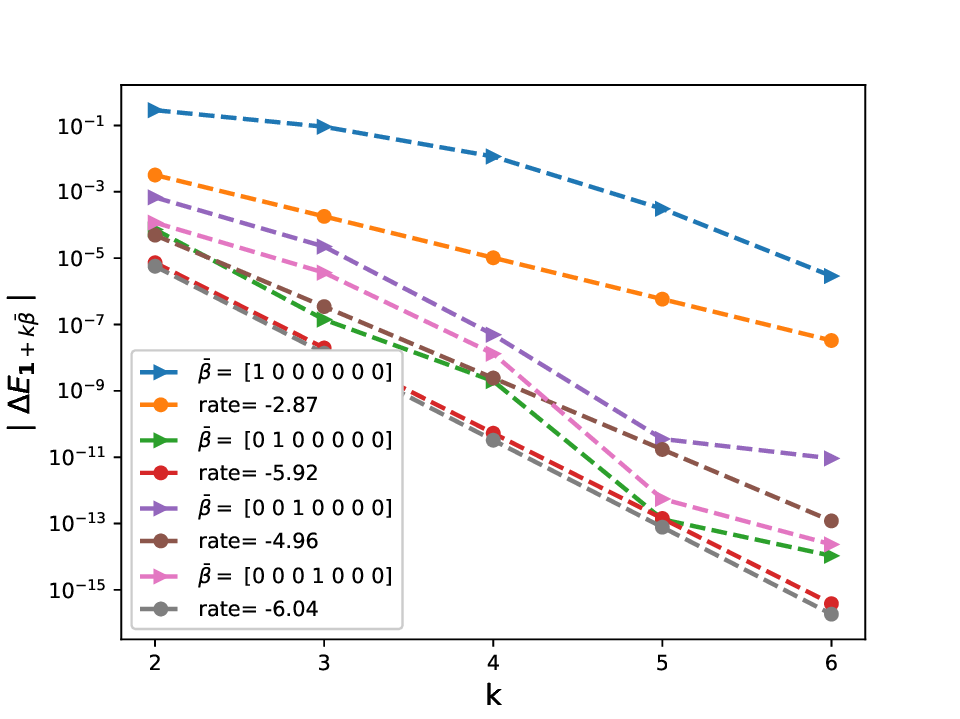}
		\caption{}
		\label{fig:4}
	\end{subfigure}
	\caption{Rate of error convergence of the first-order differences  $\abs{\Delta \text{E}_{\boldsymbol{\beta}}}$, defined in \eqref{error_contr}, ($\boldsymbol{\beta}=\mathbf{1}+k \bar{\boldsymbol{\beta}}$) for the single call option under the Heston model. The parameters are given in Set $1$ in Table  \ref{table: Examples details.}, and the number of time steps  $N=4$. We only present  the first  four dimensions used for the volatility noise (mainly $dW_v$ in \eqref{eq:dynamics Heston}). (a) Full truncation scheme, (b)  ABR scheme, and (c)  Heston OU-based scheme.}
	\label{fig:first_diff_Heston_call_N_4_set2}	
\end{figure}

\subsubsection{Comparison in terms of  the weak error behavior}\label{sec:Comparison in terms of  the weak error behavior}
We compare  the three  schemes of simulating Heston dynamics: (i) the full truncation scheme (Section \ref{sec:Discretization of Heston model with a non smooth transformations for the volatility process}),  (ii) the ABR scheme (Section \ref{sec:Euler schemes with moment matching}),  and (iii) the Heston OU-based scheme (Section \ref{sec:Discretization of Heston model with the volatility process Simulated using the sum of  Ornstein-Uhlenbeck (Bessel) processes}  in terms of the weak error convergence. We select the scheme with weak error rate of  order $1$ (\ie, $\Ordo{\Delta t}$) in the preasymptotic regime to efficiently employ the   Richardson extrapolation in our proposed methods. Figure \ref{fig:weak convergence comparison set 1}  compares  the weak error rates for the different schemes.   This figuren reveals that the  Heston OU-based scheme  exhibits a better weak convergence rate  closer to $1$ than the ABR scheme, which exhibits  a weak error rate of  $0.7$.
\begin{figure}[h!]
	\centering 
	\begin{subfigure}{0.33\textwidth}
	\includegraphics[width=\linewidth]{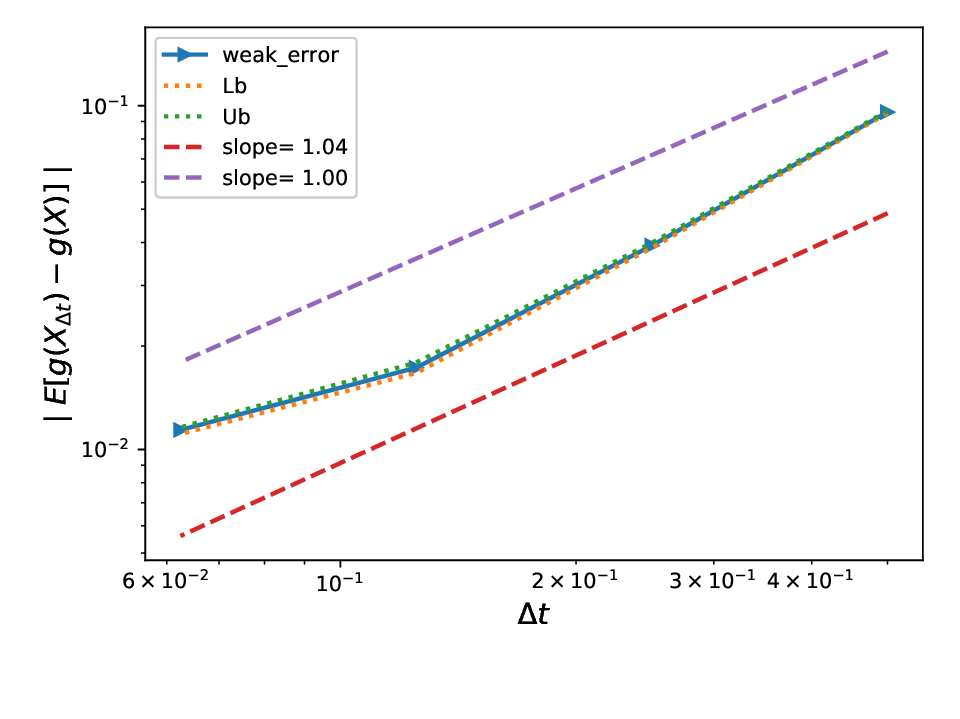}
	\caption{}
	\label{fig:weak_error_full}
\end{subfigure}\hfil 
	\begin{subfigure}{0.33\textwidth}
		\includegraphics[width=\linewidth]{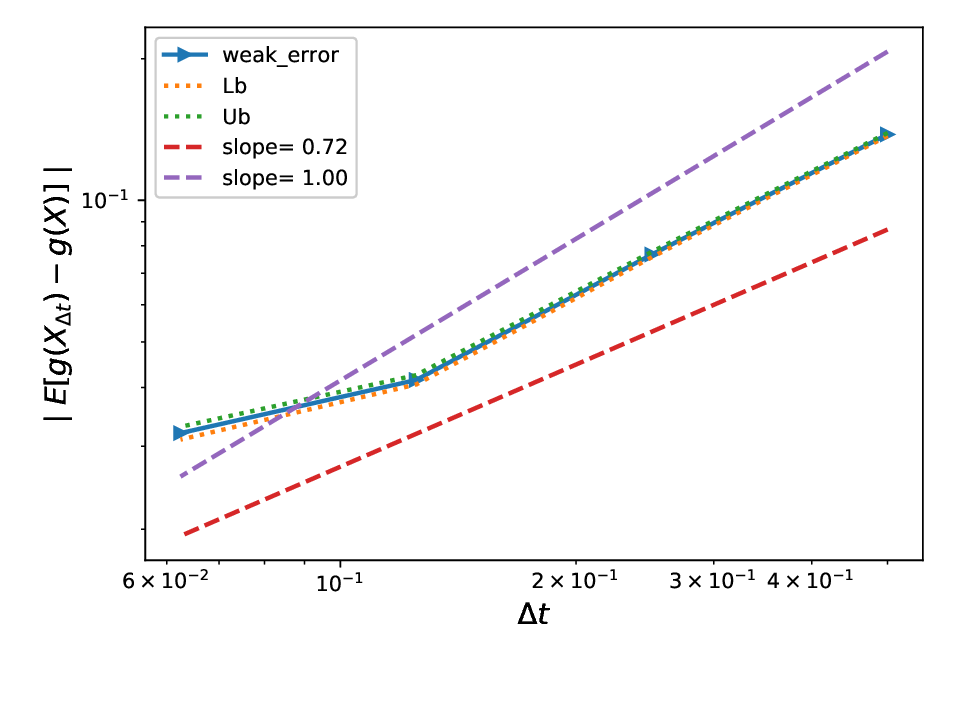}
		\caption{}
		\label{fig:weak_error_ABR}
	\end{subfigure}\hfil 
	\begin{subfigure}{0.33\textwidth}
	\includegraphics[width=\linewidth]{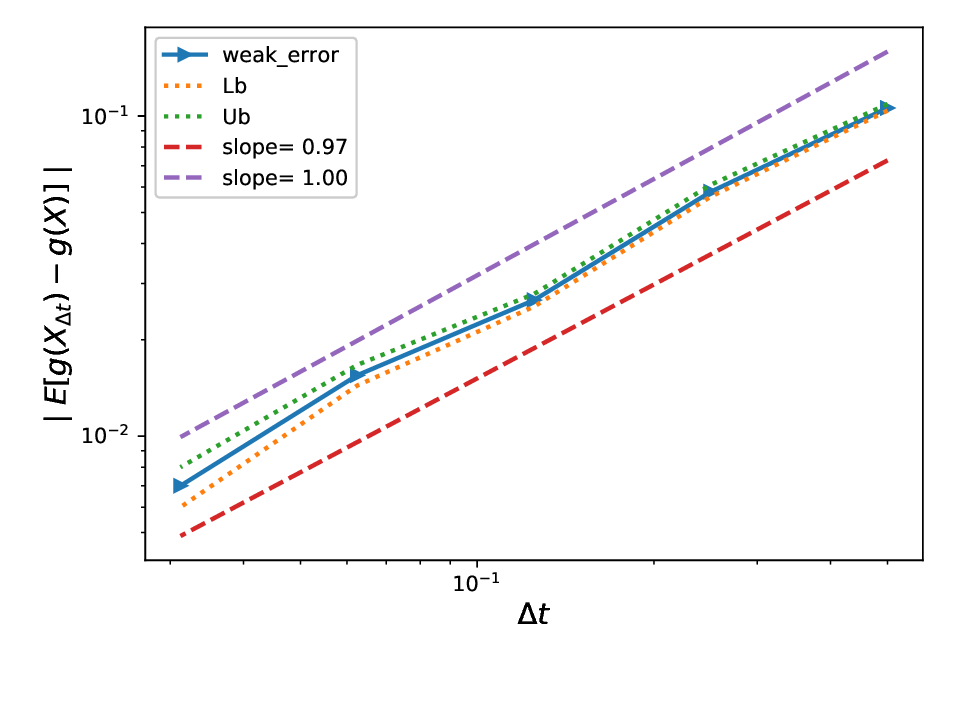}
	\caption{}
	\label{fig:1_weak_error_smooth_vol_set1}
\end{subfigure}\hfil 
	\caption{Weak error convergence for the single call option under the    Heston model for the parameters listed in Table \ref{table: Examples details.}. (a) Full truncation scheme, (b)  ABR scheme, and (c)  Heston OU-based scheme. The upper and lower bounds are $95\%$ confidence intervals. }
	\label{fig:weak convergence comparison set 1}	
\end{figure}

\end{document}